\documentclass[12pt]{amsart}
\usepackage{amsmath,amssymb,amsthm,hyperref,xcolor}
\usepackage[margin=1.25in]{geometry}
\usepackage{tikz,pgfplots}
\pgfplotsset{compat=1.17}

\newtheorem{theorem}{Theorem}[section]
\newtheorem{lemma}[theorem]{Lemma}

\theoremstyle{remark}

\newtheorem{remark}[theorem]{Remark}

\newcommand{\Trcal}{\mathcal{R}}

\newcommand{\h}{\hbar}
\newcommand{\A}{\mathcal{A}}
\newcommand{\E}{\mathcal{E}}

\title[Bathtub Potential Eigenvalues]{Asymptotic Expansion of the Eigenvalues of a Bathtub Potential with Quadratic Ends}
\author{Yuzhou Zou}
\address{Department of Mathematics\\
Northwestern University\\
2033 Sheridan Road \\
Evanston, IL 60208, U.S.A.}
\email{yuzhou.zou@northwestern.edu}
\date{\today}

\begin{document}
\maketitle

\begin{abstract}
We consider the eigenvalues of a one-dimensional semiclassical Schr\"odinger operator, where the potential consist of two quadratic ends (that is, looks like a harmonic oscillator at each infinite end), possibly with a flat region in the middle. Such a potential notably has a discontinuity in the second derivative. We derive an asymptotic expansion, valid either in the high energy regime or the semiclassical regime, with a leading order term given by the Bohr-Sommerfeld quantization condition, and an asymptotic expansion consisting of negative powers of the leading order term, with coefficients that are oscillatory in the leading order term. We apply this expansion to study the results of the Gutzwiller Trace formula and the heat kernel asymptotic for this class of potentials, giving an idea into what results to expect for such trace formulas for non-smooth potentials.
\end{abstract}

\tableofcontents

\section{Introduction}
We study the eigenvalues of the one-dimensional Schr\"odinger operator
\begin{equation}
\label{eq:ph}
P_{\h} = -\frac{\h^2}{2m}\frac{d^2}{dx^2} + V(x),\quad V(x) = \begin{cases} \frac{1}{2}m\omega_-^2x^2 & x<0 \\ 0 & 0\le x\le \ell \\ \frac{1}{2}m\omega_+^2(x-\ell)^2 & x > \ell \end{cases}.
\end{equation}
Here, $\h$ physically corresponds to the reduced Planck's constant, although we will sometimes interpret it as a small number, and ask about the \emph{semiclassical} behavior of the operator, that is, the behavior in the limit $\h\to 0$. Note that if $\ell=0$ and $\omega_-=\omega_+>0$, this corresponds to the usual harmonic oscillator. As such, here we take $m>0$, $\ell\ge 0$, and $\omega_-,\omega_+>0$, with $\omega_-\ne\omega_+$ if $\ell=0$.

\begin{figure}[h]
\centering
\begin{tikzpicture}[scale=0.7]
\begin{axis}[xmin=-3.1,xmax=5.1, ymin = -0.2, ymax = 8, axis lines = center, xlabel = {$x$}, xtick = {3}, ymajorticks=false, xticklabels = {$\ell$},
extra x ticks = {0}]
\addplot[domain=-3:0,thick,<-]{x^2/2};
\addplot[domain=0:3,thick]{0};
\addplot[domain=3:5,thick,->]{2*(x-3)^2};
\addplot[domain=-0.2:8,thick,dashed]({0},{x});
\addplot[domain=-0.2:8,thick,dashed]({3},{x});
\end{axis}
\end{tikzpicture}
\caption{The potential $V(x)$. The dashed lines are the locations of the second derivative discontinuities.}
\end{figure}
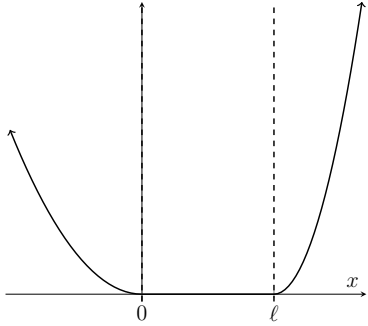
Under these assumptions, the potential $V(x)$ is $C^1$-differentiable, but there are jump(s) in the second derivative, namely at $x=0$ and $x=\ell$ (or just $x=0$ if $\ell=0$). The goal of this paper is to study the asymptotics of the eigenvalues of these operators, to give an explicit example of eigenvalue asymptotics for a potential with derivative discontinuity. Of particular interest is how these derivative discontinuity affects trace formulas, such as the Schr\"odinger trace (commonly studied through the Gutzwiller trace formula) or the heat trace.

To state the main results, we introduce a few functions. We define the angle function $\tilde\theta$ as follows: noting that since the functions $\frac{z^{1/2}}{\Gamma(3/4-z)}$ and $\frac{1}{\Gamma(1/4-z)}$ (defined on $z\ge 0$) have no common zeros, there exists a continuous angle function $\tilde\theta(z)$ satisfying
\begin{equation}
\label{eq:tildetheta}
\tilde{r}(z)\begin{pmatrix}\cos(\tilde\theta(z)) \\ \sin(\tilde\theta(z)) \end{pmatrix}  = \begin{pmatrix} \frac{z^{1/2}}{\Gamma(3/4-z)} \\ -\frac{1}{\Gamma(1/4-z)}\end{pmatrix}
\end{equation}
for some function $\tilde{r}(z)>0$. We can arrange for $\tilde\theta(0) = -\pi/2$. It turns out that $\tilde\theta$ is strictly increasing, with $\tilde\theta(z) =  \pi(z-1/4) + o(1)$ as $z\to+\infty$.
\begin{theorem}
\label{thm:eigeneq}
We have that $E$ is an eigenvalue of $P_\h$ if and only if $E>0$ and
\[\frac{\sqrt{2m}\ell\sqrt{E}}{\h} + \tilde\theta\left(\frac{E}{2\h\omega_-}\right) + \tilde\theta\left(\frac{E}{2\h\omega_+}\right)\in\pi\mathbb{Z}.\]
Moreover, each eigenvalue is simple.
Consequently, if the eigenvalues of $P_\h$ are denoted $\{E_{n;\h}\}_{n=0}^\infty$ in increasing order, then $E_{n;\h}$ satisfies the equation
\[\frac{\sqrt{2m}\ell\sqrt{E_{n;\h}}}{\h} + \tilde\theta\left(\frac{E_{n;\h}}{2\h\omega_-}\right) + \tilde\theta\left(\frac{E_{n;\h}}{2\h\omega_+}\right) = \pi n.\]
\end{theorem}
We now prove asymptotics about these eigenvalues. In the semiclassical limit, there is a connection between the behavior of the eigenvalues and the classical Hamiltonian dynamics for the symbol for the Hamiltonian $p(x,\xi) = \frac{\xi^2}{2m}+V(x)$; that is, for the dynamics of solutions to Hamilton's equation
\[\dot x(t) = \frac{\xi(t)}{m},\quad \dot\xi(t)  = -V'(x(t)).\]
For example, the Bohr-Sommerfeld quantization condition gives that the energies are quantized when the corresponding reduced action is a half-integer multiple of $2\pi \h$, i.e.\
\[S(E_{n;\h})\approx 2\pi \h(n+1/2),\]
where $S(E)$ is the reduced action
\[S(E) = \oint_{\gamma_E} \xi\,dx,\]
with the integral taken over the Hamiltonian trajectory of energy $E$.
One can compute (see Section \ref{subsec:ham}) that
\begin{equation}
\label{eq:s}
S(E)= 2\sqrt{2m}\ell\sqrt{E} + (\tau_++\tau_-)E,
\end{equation}
where
\begin{equation}
\label{eq:taupm}
\tau_{\pm} = \pi/\omega_{\pm}
\end{equation}
(note that $\tau_{\pm}$ has the interpretation of the half-period for harmonic oscillator of frequency $\omega_{\pm}$).
The Bohr-Sommerfeld quantization condition can be rephrased as $E_{n;\h}\approx \E_{n;\h}$, where
\begin{equation}
\label{eq:mathcale}
\E_{n;\h} := S^{-1}(2\pi \h(n+1/2)).
\end{equation}
We show via explicit calculation that this is indeed the case for our potentials of interest, with an asymptotic series for the remainder
\begin{equation}
\label{eq:deltae}
\Delta E_{n;\h} := E_{n;\h}-\E_{n;\h}.
\end{equation}
To state the asymptotic result, we also consider the function
\begin{equation}
\label{eq:t}
T(E) = S'(E) = \frac{\sqrt{2m}\ell}{\sqrt{E}} + (\tau_++\tau_-)E.
\end{equation}
Then $T(E)$ has the classical interpretation of the period of a classical periodic trajectory of energy $E$ (see Section \ref{subsec:ham}).

Finally, by a trigonometric polynomial we will mean a function of the form
\[P(x) = \sum_{j\in\mathbb{Z}} c_je^{ijx},\quad c_j\in\mathbb{C},\]
where only finitely many $c_j$ are non-zero;  the degree $\deg P$ is the largest value of $|j|$ for which $c_j\ne 0$. Similarly, by a two-variable trigonometric polynomial, we will mean a two-variable function $P$ which can be written as a finite sum of the form
\[P(x,y) = \sum_{j,k\in\mathbb{Z}} c_{j,k}e^{ijx}e^{iky},\quad c_{j,k}\in\mathbb{C}.\]
The degree $\deg P$ is the largest value of $|j|+|k|$ for which $c_{j,k}\ne 0$ in the sum above.

We now give the asymptotic expansion for $E_{n;\h}$. We will be interested in two regimes: either when $\h$ is fixed and $n\to\infty$ (i.e. the energy $E\to\infty$), or for a fixed range of $E$ with $\h\to 0$. In either case, we will be interested in the range of energies away from the bottom of the potential, i.e.\ at $0$; as such, we will assume that $\h<1$ and $E>\epsilon$ for some fixed $\epsilon>0$. In either regime, the quantity $\frac{\h}{E}$ will go to $0$; we can thus think of it as a ``small variable'' in the asymptotic expansion.
\begin{theorem}
\label{thm:e-asymp}
Let $\E_{n;\h}$ and $\Delta E_{n;\h}$ be as in \eqref{eq:mathcale} and \eqref{eq:deltae}. Then:
\begin{itemize}
\item For all $n\ge 1$, we have $\E_{n-1;\h}< E_{n;\h}< \E_{n+1;\h}$.
\item There exists an index set $\A\subset\mathbb{N}\times \frac{1}{2}\mathbb{N}\times\mathbb{N}$, with the property that
\[\{(j,k,l)\in\A\,:\,j< N\text{ or }k< N\}\text{ is finite for any }N,\]
and a sequence of two-variable trigonometric polynomials $\{P_{(j,k,l)}(x,y)\}_{(j,k,l)\in\A}$, with $\deg(P_{(j,k,l)})\le j/2$, such that, for a fixed $\epsilon>0$, and all $n$ satisfying $\E_{n;\h}>\epsilon$, for all $N>0$ we have
\[\frac{\Delta E_{n;\h}}{\h} = \sum_{\substack{(j,k,l)\in\A\\ j<N\text{ or }k< N}} \frac{\h^j}{\E_{n;\h}^kT(\E_{n;\h})^l}P_{j,k,l}\left(\frac{\tau_-\E_{n;\h}}{\h},\frac{\tau_+\E_{n;\h}}{\h}\right) + O\left(\frac{\h^N}{\E_{n;\h}^N}\right).\]
\end{itemize}
\end{theorem}
The first few elements of $\A$ (that is, with the lowest $j$ and $k$ values) are $(2,2,1)$, $(4,4,1)$, $(4,4,2)$, $(5,5,2)$, and $(5,11/2,3)$. The first few polynomials are given by
\begin{align*}
P_{2,2,1}(x,y) &= \frac{1}{16}(\omega_-^2\cos(x) + \omega_+^2\cos(y)), \\
P_{4,4,1}(x,y) &= -\frac{5\omega_-^4\cos(x)+5\omega_+^4\cos(y)}{128} + \frac{\omega_-^4\sin(2x)+\omega_+^4\sin(2y)}{1024}, \\
P_{4,4,2}(x,y) &= \frac{\pi}{256}(\omega_-\sin(x)+\omega_+\sin(y))(\omega_-^2\cos(x)+\omega_+^2\cos(y)), \\
P_{5,5,2}(x,y) &= \frac{1}{256}\left(\omega_-^2\cos(x)+\omega_+^2\cos(y)\right)^2, \\
P_{5,11/2,3}(x,y) &= \frac{\sqrt{m}\ell}{512\sqrt{2}}\left(\omega_-^2\cos(x)+\omega_+^2\cos(y)\right)^2.
\end{align*}
Thus, for example, taking $N = 4$, we obtain
\[E_{n;\h} = \E_{n;\h} + \frac{\h^3}{16\E_{n;\h}^2T(\E_{n;\h})}\left(\omega_-^2\cos\left(\frac{\pi\E_{n;\h}}{\h\omega_-}\right) + \omega_+^2\cos\left(\frac{\pi\E_{n;\h}}{\h\omega_+}\right)\right) + O(\h^5/\E_{n;\h}^4).\]
We comment on the history of works related to this eigenvalue expansion.
For $\ell=0$ and $\omega_-\ne\omega_+$, the result above recovers a result obtained by Razavy \cite{r-85}, who considered a two-dimensional piecewise harmonic oscillator with mismatched coefficients (the calculations are essentially one-dimensional). The author used the nearly-linear asymptotic of the one-dimensional oscillator to provide a counterexample to the Berry-Tabor conjecture regarding the distribution of the spacing between adjacent levels. More recently Chadzitaskos and Patera \cite{cp-22} (see also \cite{c-23}) analyzed the asymmetric harmonic oscillator, deriving a similar implicit equation for the eigenvalues, as well as studying the structure of the eigenfunctions to form coherent states \cite{c-24}.

In the case $\ell\ne 0$, the eigenvalues have been studied by Mazzitelli, Mazzitelli, and Soubelet \cite{mms-17} in the case $\omega_-=\omega_+$, where the authors denoted such potentials ``bathtub'' potentials. The authors derived implicit equations for the eigenvalues by explicitly solving for the coefficients in an ansatz solution, correcting some common faulty arguments encountered in solving such equations along the way (see \cite[Section 1]{mms-17} for more details). We refer the reader to \cite[Section 3]{mms-17} for some physical applications of this kind of bathtub potential.

We also note the work by Pushnitski and Sorrell \cite{ps-06}, who studied $C_c^\infty(\mathbb{R})$ perturbations to the harmonic oscillator and derived an asymptotic expansion for eigenvalues in terms of half-powers of the main asymptotic term (still given by that of the harmonic oscillator); the authors applied such an expansion to study the heat trace (as we do for our situation as well; see below) and Zeta functions constructed from the new eigenvalues.

Finally, if $\omega_-=\omega_+ = \omega$, then writing
\[V(x) - \frac{1}{2}\omega^2x^2 = \begin{cases} 0 & x<0 \\ -\frac{1}{2}\omega^2x^2 & 0\le x\le \ell \\ -\omega^2\ell x + \frac{\omega^2\ell^2}{2} & x>\ell\end{cases},\]
we can view $V(x)$ as a (singular) first-order isotropic perturbation to the harmonic oscillator. This gives a singular version of a setting studied by Doll and Zelditch \cite{dz-20}, as well as Doll, Gannot, and Wunsch \cite{dgw-18}, who studied Weyl laws and aspects of the Schr\"odinger trace for isotropic perturbations of the harmonic oscillator.

We apply these eigenvalue asymptotics to study trace formulas. In Section \ref{sec:gutz}, we study the Gutzwiller Trace Formula for these class of potentials; this is the primary motivation for deriving the main eigenvalue asymptotics. This formula ostensibly concerns the eigenvalue counting function
\[\sum_n\chi(E_{n;\h})\rho\left(\frac{E-E_{n;\h}}{\h}\right)\]
where $E$ is a fixed energy of interest, with $\chi$ and $\hat\rho$ compactly supported smooth functions, where $\hat\rho$ is the Fourier transform of $\rho$. Such a function can be interpreted as counting the number of eigenvalues near $E$; alternatively, taking $\rho$ to be roughly oscillatory, this function detects arithmetic progressions in $E_{n;\h}$ of a certain width, dependent on $\h$ and the location of $\text{supp }\hat\rho$. It turns out this eigenvalue counting function can be related to the Fourier transform of the regularized trace of a Schr\"odinger propagator, and the Gutzwiller Trace Formula subsequently relates this eigenvalue counting function to a sum dependent on the corresponding Hamiltonian dynamics, particularly concerning the periods and other dynamical quantities of periodic Hamiltonian trajectories.

To state the results in this setting, we assume from now on that $\chi\in C_c^\infty(\mathbb{R}_+)$, that is, supported in positive  energies away from zero. This is to focus away from the bottom of the potential well, where there is more complicated behavior. We also introduce the notation
\[T_{k,\alpha,\beta}(E) = kT(E) + \alpha\tau_- + \beta\tau_+\]
and
\[S_{k,\alpha,\beta}(E) = kS(E) + (\alpha\tau_-+\beta\tau_+)E\]
for $k,\alpha,\beta\in\mathbb{Z}$. Note that $T_{k,\alpha,\beta}(E) = S_{k,\alpha,\beta}'(E)$. We also consider the set
\begin{equation}
\label{eq:tne}
\mathcal{T}_N(E) = \{T_{k,\alpha,\beta}(E),:\,k,\alpha,\beta\in\mathbb{Z}, |\alpha|+|\beta|\le N\},\quad N\in\mathbb{N}.
\end{equation}
These turn out to correspond to periods of periodic trajectories which are allowed to \emph{reflect} at the singularities $x=0,\ell$ up to $N$ times; such trajectories turn out to play a role in the trace formula. A first result showing the importance of such periods is an upper bound on the semiclassical singularities of the eigenvalue-counting function:
\begin{theorem}
\label{thm:gutz-qual}
Let $\chi\in C_c^\infty(\mathbb{R}_+)$, and let $N$ be a nonnegative integer. If $\text{supp }\hat\rho\cap \mathcal{T}_N(E) = \emptyset$, then
\[\sum_{n=0}^\infty \chi(E_{n;\h})\rho\left(\frac{E-E_{n;\h}}{\h}\right) = O(\h^{2(N+1)}).\]
\end{theorem}
That is, in order for the eigenvalue-counting function to fail to be $O(\h^{2(N+1)})$, the support $\text{supp }\hat\rho$ must contain a period belonging to $\mathcal{T}_N(E)$.

In addition, for $j=1$ we show that such singularities do indeed generically appear, demonstrating that the eigenvalue-counting function is sensitive not only to the classical ``transmitted'' periodic trajectories, but also the reflected trajectories as well. Namely, we have the following:
\begin{theorem}
\label{thm:gutz-quant}
Suppose $\omega_{\pm}$, $E$ are so that the numbers
\[\tau_- = \frac{\pi}{\omega_-},\quad \tau_+ = \frac{\pi}{\omega_+},\quad T(E) = \tau_++\tau_-+\frac{\sqrt{2m}\ell}{\sqrt{E}}\]
are linearly independent over $\mathbb{Q}$. Suppose that $T^*\in\mathcal{T}_1(E)$, and suppose $\rho$ is a tempered smooth function satisfying that $(\text{supp }\hat\rho)\cap\mathcal{T}_1(E) \subseteq \{T^*\}$. Let $\chi\in C_c^\infty(\mathbb{R}_+)$. Then:
\begin{itemize}
\item If $T^* = T_{k,0,0}(E)$ for some $k\in\mathbb{Z}$, then
\[\sum_{n=0}^\infty \chi(E_{n;\h})\rho\left(\frac{E-E_{n;\h}}{\h}\right) = \frac{(-1)^ke^{iS_{k,0,0}(E)/\h}}{2\pi}\chi(E)T(E)\hat\rho(T^*) + O(\h).\]
The $O(\h)$ remainder is an asymptotic expansion in $h$, with coefficients which are linear combinations of products derivatives of $\chi$ and derivatives of $\hat\rho$ at $E$, $T^*$ respectively (in particular it is $O(\h^\infty)$ if $T^*\not\in\text{supp }\hat\rho$).
\item If $T^* = T_{k,\alpha,\beta}(E)$ for some $k,\alpha,\beta\in\mathbb{Z}$, $|\alpha|+|\beta|=1$, we have
\[\sum_{n=0}^\infty \chi(E_{n;\h})\rho\left(\frac{E-E_{n;\h}}{\h}\right) = \frac{\h^2i^{-(2k+1)}e^{iS_{k,\alpha,\beta}(E)/\h}}{64\pi E^2}\omega_{\pm}^2\chi(E)T^*\hat\rho(T^*) + O(\h^3),\]
where $\omega_{\pm}^2$ is $\omega_-^2$ if $|\alpha|=1$ and $\omega_+^2$ if $|\beta|=1$. The $O(\h^3)$ remainder is an asymptotic expansion in $\h$, with coefficients which are linear combinations of products derivatives of $\chi$ and derivatives of $\hat\rho$ at $E$, $T^*$ respectively (in particular it is $O(\h^\infty)$ if $T^*\not\in\text{supp }\hat\rho$).
\end{itemize}
\end{theorem}
We remark that the linear independence assumption is for convenience to avoid possible cancellations and could probably be relaxed with more careful bookkeeping.

Finally, in Section \ref{sec:heat}, we use the eigenvalue asymptotics to study the heat trace $\sum_{n=0}^\infty e^{-tE_{n;\h}}$ as $t\to 0^+$. For computational convenience, we will fix $\h=1$ and consider the case $\ell=0$, i.e.\ with no flat potential region, but with $\omega_-\ne\omega_+$. In this case, the potential is an asymmetric harmonic oscillator. While the leading order terms of the heat trace can be described in terms of geometric quantities, the existence of a smooth asymptotic expansion, after discarding the singular leading order term, is related to the regularity of the potential; for example, Smith showed \cite{s-19} that for a compactly supported potential $V$ on a complete manifold $M$, the difference $e^{-t(-\Delta +V)}-e^{-t(-\Delta)}$ of the perturbed and unperturbed heat kernels admits an asymptotic expansion with $m+2$ terms better than the leading singularity if and only if $V\in H^m(M)$. In our setting (which does not belong to the setting in \cite{s-19} as the potentials here are quadratic at infinity), we show that the derivative discontinuity in this case is strong enough to preclude an infinite asymptotic expansion, as follows:
\begin{theorem}
\label{thm:heat}
Suppose $\ell=0$ and $\omega_-\ne\omega_+$. Let $\overline\omega = 2/(\omega_-^{-1}+\omega_+^{-1})$ be the harmonic mean of $\omega_-$ and $\omega_+$, and write $E_n = E_{n;1}$. Then there exist constants $c_1,\dots,c_4$ such that
\[
\sum_{n=0}^\infty e^{-tE_n} 
=\frac{1}{\overline\omega t} +\sum_{j=1}^4 c_jt^j + \frac{(\omega_+^2-\omega_-^2)^2}{4096\overline\omega T^2}t^5\log(t)+ O(t^5).
\]
In particular, $\sum_{n=0}^\infty e^{-tE_n} - \frac{1}{\overline\omega t}$ is $C^4$ but not $C^5$ as $t\to 0^+$.
\end{theorem}
We remark that for this case ($\ell=0$), we can also derive a semiclassical asymptotic expansion, essentially by replacing $t$ by $\h t$ in the above expansion. A semiclassical heat trace expansion was studied in \cite{guw-12} for a (semiclassical, i.e.\ order $\h^2$) perturbation of the harmonic oscillator; a key qualitative difference here is that the expansion here would contain terms like $\log\h$, which did not appear in \cite{guw-12}.

\subsection*{Acknowledgments}
The author would like to thank David Sher and Jared Wunsch for useful discussions on this project.

\section{Equation for eigenvalues}
\label{sec:eigeneq}
In this section we derive the equation satisfied by the eigenvalues $E_{n;\h}$; that is, we prove Theorem \ref{thm:eigeneq}.

We recall \cite{as} for $a\in\mathbb{C}$ that the solutions to the so-called Weber differential equation
\[-u''(x) + \left(\frac{1}{4}x^2 + a\right)u(x) = 0\]
are given by the \emph{parabolic cylinder functions}, with the subspace of solutions which are also in $L^2((0,\infty))$ having dimension one, spanned by the function denoted $U(a,x)$ in \cite{as} (sometimes called \emph{Whittaker's function}). This function has the properties that
\begin{equation}
\label{eq:u-values}
U(a,0) = \frac{\sqrt{\pi}}{2^{a/2+1/4}\Gamma(3/4+a/2)},\quad U'(a,0) = -\frac{\sqrt{\pi}}{2^{a/2-1/4}\Gamma(1/4+a/2)}
\end{equation}
(cf. \cite[Equation 19.3.5]{as}). In particular, recalling the function $\tilde\theta$ defined in \eqref{eq:tildetheta}, note for $a>0$ that
\begin{equation}
\label{eq:uu'}
\begin{pmatrix} U(-a,0) \\ \frac{1}{\sqrt{a}}U'(-a,0) \end{pmatrix} = \frac{\sqrt{\pi}}{2^{-a/2-1/4}}\begin{pmatrix} \frac{\sqrt{a/2}}{\Gamma(3/4-a/2)} \\ -\frac{1}{\Gamma(1/4-a/2)}\end{pmatrix} =\frac{\sqrt{\pi}r(a/2)}{2^{-a/2-1/4}} \begin{pmatrix} \cos(\tilde\theta(a/2)) \\ \sin(\tilde\theta(a/2))\end{pmatrix},
\end{equation}
so the vector $\begin{pmatrix} U(-a,0) \\ \frac{1}{\sqrt{a}}U'(-a,0) \end{pmatrix}$ is a positive multiple of $\begin{pmatrix} \cos(\tilde\theta(a/2)) \\ \sin(\tilde\theta(a/2))\end{pmatrix}$, i.e.\ makes an angle of $\tilde\theta(a/2)$ with respect to the vector $\begin{pmatrix} 1 \\ 0 \end{pmatrix}$.

We now ask when the equation $P_\h u = Eu$ has a nonzero $L^2$ solution. Since $V(x)$ is continuous, any $L^2$ solution $u$ must belong in $H^2_{loc}(\mathbb{R})$, and in particular must be $C^1$. Moreover, it must solve the equations
\begin{align*}
-\frac{\h^2}{2m}\frac{d^2u}{dx^2} + \frac{1}{2}m\omega_-x^2u &= Eu\text{ for }x<0,\\
-\frac{\h^2}{2m}\frac{d^2u}{dx^2} &= Eu\text{ for }0< x<\ell,\\
-\frac{\h^2}{2m}\frac{d^2u}{dx^2} + \frac{1}{2}m\omega_+(x-\ell)^2u &= Eu\text{ for }x>\ell,
\end{align*}
with $u$ also being $L^2$ in each of the above intervals.

In the region $x<0$, we can express $u$ in terms of the distinguished solution $\tilde{U}$ of the Weber differential equation. Indeed, making the change of variables (with $\gamma>0$)
\[z=-\gamma x \implies \frac{d}{dx} = -\gamma\frac{d}{dz},\]
it follows that $v(z) = u(-z/\gamma)$ is a $L^2((0,\infty))$ solution of
\[-\frac{\h^2\gamma^2}{2m}\frac{d^2v}{dz^2} + \frac{1}{2}\frac{m\omega_-^2}{\gamma^2}z^2v = Ev,
\quad\text{i.e.}\quad
-\frac{d^2v}{dz^2} + \left(\frac{m^2\omega_-^2}{\h^2\gamma^4}z^2 - \frac{2Em}{\h^2\gamma^2}\right)v = 0.\]
Setting $\frac{m^2\omega_-^2}{\h^2\gamma^4} = \frac{1}{4}$, i.e.\ $\gamma = \sqrt{\frac{2m\omega_-}{\h}}$, it follows that
\[-\frac{d^2v}{dz^2} + \left(\frac{1}{4}z^2-a_-\right)v = 0,\quad a_- = \frac{E}{\h\omega_-},\]
and hence $v(z)$ is a multiple of $\tilde{U}(-a_-,z)$ with $a_- = E/(\h\omega_-)$. This implies that
\begin{equation}
\label{eq:ule0}
u(x) = v(-\gamma x) = C_-\tilde{U}\left(-\frac{E}{\h\omega_-},-\sqrt{\frac{2m\omega_-}{\h}}x\right),\quad x\le 0.
\end{equation}
Similar logic gives
\begin{equation}
\label{eq:ugeell}
u(x) = C_+\tilde{U}\left(-\frac{E}{\h\omega_+},\sqrt{\frac{2m\omega_+}{\h}}(x-\ell)\right),\quad x\ge\ell.
\end{equation}
Finally, for $0\le x\le\ell$, we have that $u$ solves $-\frac{\h^2}{2m}\frac{d^2u}{dx^2} = Eu$, and hence
\[u(x) = \cos\left(\frac{\sqrt{2mE}}{\h}x\right)u(0) + \frac{\sin\left(\frac{\sqrt{2mE}}{\h}x\right)}{\sqrt{2mE}/\h}u'(0)\]
for $0\le x\le\ell$ (note that $u$ being $C^1$ means we may freely write the values of $u$ and $u'$ instead of taking limits). This could equivalently be written as
\begin{equation}
\label{eq:urotate}
\begin{pmatrix} u(x) \\ -\frac{\h}{\sqrt{2mE}}u'(x)\end{pmatrix} = \begin{pmatrix} \cos\left(\frac{\sqrt{2mE}}{\h}x\right) & -\sin\left(\frac{\sqrt{2mE}}{\h}x\right) \\ \sin\left(\frac{\sqrt{2mE}}{\h}x\right) & \cos\left(\frac{\sqrt{2mE}}{\h}x\right)\end{pmatrix}\begin{pmatrix} u(0) \\ -\frac{\h}{\sqrt{2mE}}u'(0)\end{pmatrix}.
\end{equation}
If we let $\gamma_{\pm} = \sqrt{\frac{2m\omega_{\pm}}{\h}}$, and note that $\frac{\h}{\sqrt{2mE}}\gamma_{\pm} = \sqrt{\frac{\h\omega_{\pm}}{E}} = \frac{1}{\sqrt{a_{\pm}}}$, it follows that
\begin{equation}
\label{eq:u(0)}
\begin{aligned}
\begin{pmatrix} u(0) \\ -\frac{\h}{\sqrt{2mE}}u'(0)\end{pmatrix} = C_-\begin{pmatrix} U(-a_-,0) \\ -\frac{\h}{\sqrt{2mE}}(-\gamma_-)U'(-a_-,0)\end{pmatrix} &= C_-\begin{pmatrix} U(-a_-,0) \\ \frac{1}{\sqrt{a_-}}U'(-a_-,0)\end{pmatrix} \\
&= \tilde{C}_-\begin{pmatrix} \cos(\tilde\theta(a_-/2)) \\ \sin(\tilde\theta(a_-/2))\end{pmatrix}.
\end{aligned}
\end{equation}
i.e. the vector $\begin{pmatrix} u(0) \\ -\frac{\h}{\sqrt{2mE}}u'(0)\end{pmatrix}$ has angle $\tilde\theta(a_-/2)$ modulo $\pi$. Similarly,
\begin{equation}
\label{eq:u(ell)}
\begin{aligned}
\begin{pmatrix} u(\ell) \\ -\frac{\h}{\sqrt{2mE}}u'(\ell)\end{pmatrix} &= C_+\begin{pmatrix} U(-a_+,0) \\ -\frac{1}{\sqrt{a_+}}U'(-a_+,0)\end{pmatrix} \\
&= \tilde{C}_+\begin{pmatrix} \cos(\tilde\theta(a_+/2)) \\ -\sin(\tilde\theta(a_+/2))\end{pmatrix} = \tilde{C}_+\begin{pmatrix} \cos(-\tilde\theta(a_+/2)) \\ \sin(-\tilde\theta(a_+/2))\end{pmatrix}.
\end{aligned}
\end{equation}
Thus, applying \eqref{eq:urotate} to $x=\ell$, and considering angles using \eqref{eq:u(0)} and \eqref{eq:u(ell)}, we obtain
\[-\tilde\theta(a_+/2) = \frac{\sqrt{2mE}}{\h}\ell + \tilde\theta(a_-/2)\,\mod\pi,\quad
\text{i.e.}
\quad\tilde\theta(a_-/2) + \tilde\theta(a_+/2) + \frac{\sqrt{2mE}\ell}{\h}\in\pi\mathbb{Z}.\]
This proves Theorem \ref{thm:eigeneq}.
\begin{remark}
Consider the case $\omega_-=\omega_+=\omega$ and $\ell\ne 0$, considered in \cite{mms-17}.
In this case, the equation simplifies to
\[2\tilde\theta\left(\frac{E_{n;\h}}{2\h\omega}\right) + \frac{\sqrt{2m}\ell\sqrt{E_{n;\h}}}{\h} = \pi n.\]
Defining the dimensionless quantities $\E = E_{n;\h}/(\h\omega)$ and $l = \sqrt{\frac{m\omega}{2\h}}\ell$, we obtain
\[\tilde\theta\left(\frac{\E}{2}\right) + l\sqrt{\E} = \frac{\pi n}{2},\]
and hence
\[-\frac{\Gamma(3/4-\E/2)}{(\E/2)^{1/2}\Gamma(1/4-\E/2)} = \tan\left(\tilde\theta\left(\frac{\E}{2}\right)\right) = \tan\left(\frac{\pi n}{2} - l\sqrt{\E}\right).\]
Since
\[\tan\left(\frac{\pi n}{2}-x\right) = \begin{cases} \cot(x) & n\text{ odd }\\ -\tan(x) & n\text{ even }\end{cases},\]
we thus obtain the equation
\[-\sqrt{2}\frac{\Gamma(3/4-\E/2)}{\Gamma(1/4-\E/2)} = \begin{cases} \sqrt{\E}\cot(l\sqrt{\E}) & n\text{ odd } \\ -\sqrt{\E}\tan(l\sqrt{\E}) & n\text{ even }\end{cases}.\]
This recovers Equations (33) and (34) in \cite{mms-17}.
\end{remark}

\section{Eigenvalue Asymptotics}
\label{sec:e-asymp}
In this section we derive the eigenvalue asymptotics claimed in Theorem \ref{thm:e-asymp}.
We start by looking at the asymptotics of the angle function $\tilde\theta$:
\begin{lemma}
\label{lem:tildethetaasymp}
The function $\tilde\theta$ defined in \eqref{eq:tildetheta} satisfying $\tilde\theta(0)=-\pi/2$ is strictly increasing on $[0,\infty)$. Moreover, it can be written as
\[\tilde\theta(z) = \pi\left(z-\frac{1}{4}\right) + r(z),\]
where $r(z)$ is smooth for $z>0$, $|r(z)|<\pi/2$ for all $z>0$, and as $z\to\infty$, $r(z)$ satisfies an asymptotic expansion of the form
\[r(z) = \sum_{j=2}^N \frac{P_j(2\pi z)}{z^j} + O(z^{-(N+1)}),\]
where $P_j$ is a trigonometric polynomial of degree at most $j/2$.
\end{lemma}
The proof of Lemma \ref{lem:tildethetaasymp} is given in Appendix \ref{sec:lemma-proofs}.
In fact, the expansion up to $N=5$ for $r(z)$ is given by
\[r(z) = -\frac{\cos(2\pi z)}{128z^2} + \frac{5\cos(2\pi z)}{4096z^4} - \frac{\sin(4\pi z)}{32768z^4} + O(z^{-6}).\]
Theorem \ref{thm:eigeneq} then becomes
\[\frac{\sqrt{2m}\ell\sqrt{E_{n;\h}}}{\h} + \frac{\pi E_{n;\h}}{2\h\omega_-} + \frac{\pi E_{n;\h}}{2\h\omega_+} - \frac{\pi}{2} + r\left(\frac{E_{n;\h}}{2\h\omega_-}\right) + r\left(\frac{E_{n;\h}}{2\h\omega_+}\right) = \pi n.\]
Recalling that $\tau_{\pm} = \pi/\omega_{\pm}$ and $S(E)= 2\sqrt{2m}\ell\sqrt{E} + (\tau_++\tau_-)E$,
we can rewrite the result of Theorem \ref{thm:eigeneq} as the equation
\begin{equation}
\label{eq:s+r}
\frac{S(E_{n;\h})}{2\h} +  r\left(\frac{E_{n;\h}}{2\h\omega_-}\right) + r\left(\frac{E_{n;\h}}{2\h\omega_+}\right) = \pi\left(n+\frac{1}{2}\right).
\end{equation}
Note that in the regime $E/\h\to\infty$ (i.e.\ either if $\h$ is fixed and $n\to\infty$, or $E$ is in a fixed positive window but $\h\to 0$), we have that the $r$ terms above go to zero, leaving an approximate equation
\[\frac{S(E_{n;\h})}{2\h}\approx \pi\left(n+\frac{1}{2}\right),\quad i.e.\ S(E_{n;\h})\approx 2\pi \h\left(n+\frac{1}{2}\right).\]
This suggests the Bohr-Sommerfeld approximation $E_{n;\h}\approx\E_{n;\h}$, where $\E_{n;\h}$ defined in \eqref{eq:mathcale} satisfies $S(\E_{n;\h}) = 2\pi \h(n+1/2)$.

We now use Equation \eqref{eq:s+r} to derive the eigenvalue asymptotics, in particular giving an estimate to how close $E_{n;\h}$ is to $\E_{n;\h}$.
\begin{proof}[Proof of Theorem \ref{thm:e-asymp}]
We first prove the property $\E_{n-1;\h}<E_{n;\h}<\E_{n+1;\h}$; this gives an upper bound on how far $E_{n;\h}$ can deviate from $\E_{n;\h}$. Note that
\[\frac{S(\E_{n-1;\h})}{2\h} + r\left(\frac{\E_{n-1;\h}}{2\h\omega_-}\right) + r\left(\frac{\E_{n-1;\h}}{2\h\omega_+}\right) < \pi(n-1/2) + \pi/2 + \pi/2 = \pi(n+1/2).\]
The function $\frac{S(E)}{2\h} + r\left(\frac{E}{2\h\omega_-}\right) + r\left(\frac{E}{2\h\omega_+}\right)$ is monotonic in $E$, since it is equal to $\frac{\sqrt{2m}\ell\sqrt{E}}{\h} + \tilde\theta\left(\frac{E}{2\h\omega_-}\right) + \tilde\theta\left(\frac{E}{2\h\omega_+}\right)$, whose monotonicity follows from Lemma \ref{lem:tildethetaasymp}. Thus, by monotonicity and Equation \eqref{eq:s+r}, we must have $\E_{n-1;\h}<E_{n;\h}$. A similar argument shows that $E_{n;\h}<\E_{n+1;\h}$.

Using that $\E_{n;\h} = S^{-1}(2\pi \h (n+1/2))$, and that $S^{-1}$ is smooth with bounded derivative\footnote{Indeed, we can explicitly solve
\[S^{-1}(w) = \frac{w}{\tau_++\tau_-} + 4m\ell^2 - 2\sqrt{\left(\frac{2m\ell^2}{(\tau_++\tau_-)^2}\right)^2+\frac{2m\ell^2w}{(\tau_++\tau_-)^3}};\]
note that the term inside the square root is either strictly positive near $w=0$ (if $\ell>0$) or identically zero (if $\ell=0$). Thus $S^{-1}$ is smooth near $0$; combining this with the observation that $\lim\limits_{w\to+\infty}(S^{-1})'(w) = \frac{1}{\tau_++\tau_-}$, we see that $(S^{-1})'$ is bounded.} on $(0,\infty)$ (in particular Lipschitz), we see that
\[\E_{n+1;\h}-\E_{n;\h} = S^{-1}(2\pi \h(n+3/2))-S^{-1}(2\pi \h(n+1/2)) \le 2\pi C\h,\]
where $C$ is the Lipschitz constant of $S^{-1}$. Thus $E_{n;\h}\in (\E_{n-1;\h},\E_{n+1;\h})$ implies
$\Delta E_{n;\h} = O(\h)$.

We now use \eqref{eq:s+r} as follows. By expanding $S$ as a Taylor series up to order $1$, we see that
\begin{equation}
\label{eq:s-taylor1}
S(E_{n;\h}) = S(\E_{n;\h}) + S'(\E_{n;\h})\Delta E_{n;\h} + \frac{S''(E^*)}{2}(\Delta E_{n;\h})^2
\end{equation}
for some $E^*$ between $\E_{n;\h}$ and $E_{n;\h}$. Since $S''(E) = O(E^{-3/2})$ and $E_{n;\h}/\E_{n;\h} = O(1)$ either as $n\to\infty$ or $\h\to 0$, we have that $S''(E^*) = O(\E_{n;\h}^{-3/2})$. Combined with $\Delta E_{n;\h} = O(\h)$, we thus have
\[\frac{S(E_{n;\h})}{2\h} = \frac{S(\E_{n;\h})}{2\h} + \frac{S'(\E_{n;\h})}{2}\frac{\Delta E_{n;\h}}{\h} + O(\h\E_{n;\h}^{-3/2}).\]
Moreover, $r(z) = O(z^{-2})\implies r\left(\frac{\E_{n;\h}}{2\h\omega_{\pm}}\right) = O(\h^2\E_{n;\h}^{-2})$. Since $S(\E_{n;\h}) = 2\pi \h(n+1/2)$, \eqref{eq:s+r} becomes
\begin{equation}
\label{eq:deltae-inter}
\frac{S'(\E_{n;\h})}{2}\frac{\Delta E_{n;\h}}{\h} + O(\h\E_{n;\h}^{-3/2}) + O(\h^2\E_{n;\h}^{-2}) = 0\implies \frac{\Delta E_{n;\h}}{\h} = O(\h\E_{n;\h}^{-3/2})
\end{equation}
since $S'(\E_{n;\h})$ is bounded from above and below in the region of interest.

We now expand the asymptotic series for $r$. Since
\[r(z) = -\frac{\cos(2\pi z)}{128 z^2} + O(z^4), \quad r'(z)= O(z^{-2}),\]
it follows that
\[r\left(\frac{E_{n;\h}}{2\h\omega_{\pm}}\right) = r\left(\frac{\E_{n;\h}}{2\h\omega_{\pm}}\right)  + r'\left(\frac{E^*}{2\h\omega_{\pm}}\right)\frac{\Delta E_{n;\h}}{2\h\omega_{\pm}}\]
for some $E^*$ between $E_{n;\h}$ and $\E_{n;\h}$. The first term equals
\[-\frac{\h^2\omega_{\pm}^2\cos\left(\frac{\pi\E_{n;\h}}{\omega_{\pm}\h}\right)}{32\E_{n;\h}^2} + O(\h^4\E_{n;\h}^{-4}),\]
while the second term is the product of the derivative term which is $O(\h^2\E_{n;\h}^{-2})$ (again $E^*$ equals $\E_{n;\h}$ up to a bounded multiplicative factor) and the $\frac{\Delta E_{n;\h}}{\h}$ term, which is $O(\h\E_{n;\h}^{-3/2})$ as just shown. This gives
\[r\left(\frac{E_{n;\h}}{2\h\omega_{\pm}}\right) = -\frac{\h^2\omega_{\pm}^2\cos\left(\frac{\pi\E_{n;\h}}{\omega_{\pm}\h}\right)}{32\E_{n;\h}^2} + O(\h^3\E_{n;\h}^{-7/2}).\]
Furthermore, we can rewrite \eqref{eq:s-taylor1} as
\[\frac{S(E_{n;\h})}{2\h} = \frac{S(\E_{n;\h})}{2\h} + \frac{S'(\E_{n;\h})}{2}\frac{\Delta E_{n;\h}}{\h} + \frac{\h S''(E^*)}{4}\left(\frac{\Delta E_{n;\h}}{\h}\right)^2;\]
the last term is now $O(\h^3\E_{n;\h}^{-9/2})$ using the improved estimate \eqref{eq:deltae-inter}. It follows that
\[\frac{S'(\E_{n;\h})}{2}\frac{\Delta E_{n;\h}}{\h} + O(\h^3\E_{n;\h}^{-9/2})  -\frac{\h^2\omega_-^2\cos\left(\frac{\pi\E_{n;\h}}{\omega_-\h}\right)}{32\E_{n;\h}^2} -\frac{\h^2\omega_+^2\cos\left(\frac{\pi\E_{n;\h}}{\omega_{+}\h}\right)}{32\E_{n;\h}^2} + O(\h^3\E_{n;\h}^{-7/2})=0.\]
Recalling $T(E) = S'(E)$, this thus gives
\[\frac{\Delta E_{n;\h}}{\h} = \frac{\h^2}{16\E_{n;\h}^2T(\E_{n;\h})}\left(\omega_-^2\cos\left(\frac{\tau_-\E_{n;\h}}{\h}\right) + \omega_+^2\cos\left(\frac{\tau_+\E_{n;\h}}{\h}\right)\right) + O(\h^3\E_{n;\h}^{-7/2}).\]
This gives the leading order term in the error $\frac{\Delta E_{n;\h}}{\h}$.

The remaining terms can be calculated similarly, by expanding the Taylor/asymptotic series for $S$ and $r$ appropriately. We prove the desired form inductively. 

To define our index set $\A$, we let it be the smallest subset of $\mathbb{N}\times \frac{1}{2}\mathbb{N}\times\mathbb{N}$, containing $\{(2k,2k,1)\}_{k=1,2,\dots}$, which is closed under the following two operations:
\begin{itemize}
\item If $(j_1,k_1,l_1)$, $(j_2,k_2,l_2)$, \dots, $(j_m,k_m,l_m)$ are $m$ (not necessarily distinct) elements of $\A$, then for all $n\ge 2$ we have
\[\left(\sum_{i=1}^m j_i + n, \sum_{i=1}^m k_i + n, \sum_{i=1}^m l_i + 1\right)\in\A.\]
\item If $(j_1,k_1,l_1)$, $(j_2,k_2,l_2)$, \dots, $(j_m,k_m,l_m)$ are $m$ (not necessarily distinct) elements of $\A$, with $m\ge 2$, then
\[\left(\sum_{i=1}^m j_i + (m-1), \sum_{i=1}^m k_i + (m-1/2), \sum_{i=1}^m l_i + 1\right)\in\A.\]
\end{itemize}
This set can be constructed by starting with $\A_0 = \{(2k,2k,1)\}_{k=1,2,\dots}$, and, for each $i$, let $\A_{i+1}$ consist of all tuples obtained by applying the above operations to finitely many elements (with repetition) of $\A_i$. One can verify\footnote{
For example, to check that the condition $k\le\frac{7j-2}{6}$ is preserved by the two operations, note that if
\[j^* = \sum_{i=1}^m j_i + n, k^* = \sum_{i=1}^m k_i+n,\quad n\ge 2,\]
and $k_i\le (7j_i-2)/6$, then
\[k^* \le \sum_{i=1}^m\left(\frac{7j_i-2}{6}\right) + n = \frac{7}{6}\left(\sum_{i=1}^m j_i + n\right) - \frac{2m+n}{6} \le \frac{7}{6}k^* - \frac{2}{6}\]
since $2m+n>n\ge 2$. On the other hand, if
\[j^* = \sum_{i=1}^m j_i + (m-1),\quad k^* = \sum_{i=1}^m k_i + (m-1/2),\quad m\ge 2,\]
then
\[k^* \le \sum_{i=1}^m\frac{7j_i-2}{6} + (m-1/2) = \frac{7}{6}\left(\sum_{i=1}^m j_i + (m-1)\right) - \frac{m}{3} - \frac{7(m-1)}{6} + m-\frac{1}{2} \le\frac{7}{6}j^* - \frac{2}{6},\]
the last line inequality following since
\[\frac{m}{3}+\frac{7(m-1)}{6}-m+\frac{1}{2} = \frac{3m-4}{6}\ge\frac{2}{6}\]
since $m\ge 2$.
} that the set
\[\left\{(j,k,l)\in\mathbb{N}\times\frac{1}{2}\mathbb{N}\times\mathbb{N}\,:\,j\ge 2, j\le k\le \frac{7j-2}{6}, l\le \frac{2j-1}{3}\right\}\]
is preserved by the above operations and contains $\{(2k,2k,1)\}_{k=1,2,\dots}$; hence, we have that $\A$ is contained in the above set; such a set also satisfies that its intersection with $\{(j,k,l)\,:\,j<N\text{ or }k<N\}$ is finite for any $N$.

We can then proceed inductively. The case $N=3$ is proven above. We now suppose a desired expansion exists up to order $N$. We now Taylor expand $S$ as
\begin{align*}
\frac{S(E_{n;\h})}{2\h} &= \frac{S(\E_{n;\h})}{2\h} + \frac{T(\E_{n;\h})}{2}\frac{\Delta E_{n;\h}}{\h} \\
&+ \sum_{j=2}^{N-1} \frac{\h^{j-1}S^{(j)}(\E_{n;\h})}{2j!}\left(\frac{\Delta E_{n;\h}}{\h}\right)^j + \frac{\h^{N-1}S^{(N)}(E^*)}{2N!}\left(\frac{\Delta E_{n;\h}}{\h}\right)^N.
\end{align*}
Using that $\frac{\Delta E_{n;\h}}{\h} = O(\h^2/\E_{n;\h}^2)$, the last term is $O(\h^{3N-1}\E_{n;\h}^{-(3N-1/2)})$ using similar arguments as above. By similar considerations,
\[r\left(\frac{E_{n;\h}}{2\h\omega_{\pm}}\right) = \sum_{j=0}^{N-1} \frac{r^{(j)}\left(\frac{\E_{n;\h}}{2\h\omega_{\pm}}\right)}{j!(2\omega_{\pm})^j}\left(\frac{\Delta E_{n;\h}}{\h}\right)^j + O\left(\h^{2(N+1)}\E_{n;\h}^{-2(N+1)}\right)\]
since $r^{(j)}(z) = O(z^{-2})$ for each $j$.

Using these Taylor expansions, \eqref{eq:s+r} becomes
\begin{equation}
\label{eq:delta-e-recursive}
\begin{aligned}\frac{\Delta E_{n;\h}}{\h} &= -\frac{2}{T(\E_{n;\h})}\left(\sum_{j=2}^{N-1} \frac{\h^{j-1}S^{(j)}(\E_{n;\h})}{2j!}\left(\frac{\Delta E_{n;\h}}{\h}\right)^j\right.\\
&\left. + \sum_{\pm}\sum_{j=0}^{N-1} \frac{r^{(j)}\left(\frac{\E_{n;\h}}{2\h\omega_\pm}\right)}{j!(2\omega_\pm)^j}\left(\frac{\Delta E_{n;\h}}{\h}\right)^j  + O(\h^{2(N+1)}\E_{n;\h}^{-2(N+1)})\right).
\end{aligned}
\end{equation}
We now take the inductive assumption
\begin{equation}
\label{eq:deltae-inductive}
\frac{\Delta E_{n;\h}}{\h} = \sum_{\substack{(j,k,l)\in\A\\ j<N\text{ or }k< N}} \frac{\h^j}{\E_{n;\h}^kT(\E_{n;\h})^l}P_{(j,k,l)}\left(\frac{\tau_-\E_{n;\h}}{\h},\frac{\tau_+\E_{n;\h}}{\h}\right) + r_N
\end{equation}
with $r_N = O(\h^N/\E_{n;\h}^N)$ and plug it into the above equation.

We look at terms of the form
\[\frac{\h^{j-1}S^{(j)}(\E_{n;\h})}{T(\E_{n;\h})}\left(\frac{\Delta E_{n;\h}}{\h}\right)^j,\quad j\ge 2\]
by expanding the power $\left(\frac{\Delta E_{n;\h}}{\h}\right)^j$ using the sum in \eqref{eq:deltae-inductive}. If any of the terms are multiplied by $r_N$, then the resulting term is $O(\h^N/\E_{n;\h}^N)$ times a product of at least one term which is $O(\h/\E_{n;\h})$; hence the overall product will be $O(\h^{N+1}/\E_{n;\h}^{N+1})$. We can thus ignore products of terms times $r_N$. The remaining terms will be products of the form
\[\frac{\h^{j-1}C_j}{\E_{n;\h}^{j-1/2}T(\E_{n;\h})}\prod_{i=1}^j \frac{\h^{j_i}}{\E_{n;\h}^{k_i}T(\E_{n;\h})^{l_i}}P_{j_i,k_i,l_i}\left(\frac{\tau_-\E_{n;\h}}{\h},\frac{\tau_+\E_{n;\h}}{\h}\right)\]
with $(j_i,k_i,l_i)\in\A$, $\deg P_{j_i,k_i,l_i}\le j_i/2$. This can be rewritten as
\[\frac{\h^{j^*}}{\E_{n;\h}^{k^*}T(\E_{n;\h})^{l^*}}\tilde{P}\left(\frac{\tau_-\E_{n;\h}}{\h},\frac{\tau_+\E_{n;\h}}{\h}\right)\]
where
\[j^* = \sum_{i=1}^j j_i + (j-1),\quad k^* = \sum_{i=1}^j k_i + (j-1/2),\quad l^* = \sum_{i=1}^j l_i + 1,\]
and $\tilde{P}=\prod_{i=1}^j P_{(j_i,k_i,l_i)}$; note that $\deg \tilde{P}\le\sum_{i=1}^j\deg P_{(j_i,k_i,l_i)}\le (\sum_{i=1}^j j_i)/2\le j^*/2$. By construction of $\A$, we have $(j^*,k^*,l^*)\in\A$; hence this term is of the form claimed.

Similarly, we analyze terms of the form
\[\frac{r^{(j)}\left(\frac{\E_{n;\h}}{2\h\omega_\pm}\right)}{j!(2\omega_\pm)^j}\left(\frac{\Delta E_{n;\h}}{\h}\right)^j.\]
For each $j$, we have that $r^{(j)}(z)$ is a finite sum of terms of the form $\tilde{P}_k(2\pi z)/z^k$, with $k\ge 2$ and $\deg\tilde{P}_k\le k/2$. Thus expanding the power gives terms of the form
\[\frac{\h^k\tilde{P}_k\left(\frac{\pi\E_{n;\h}}{\omega_\pm \h}\right)}{\E_{n;\h}^kT(\E_{n;\h})}\prod_{i=1}^j \frac{\h^{j_i}}{\E_{n;\h}^{k_i}T(\E_{n;\h})^{l_i}}P_{(j_i,k_i,l_i)}\left(\frac{\tau_-\E_{n;\h}}{\h},\frac{\tau_+\E_{n;\h}}{\h}\right).\]
This can be written as $\frac{\h^{j^*}}{\E_{n;\h}^{k^*}T(\E_{n;\h})^{l^*}}\tilde{P}$ with $j^* = \sum_{i=1}^j j_i + k$, $k^* = \sum_{i=1}^j k_i + k$, $l^* = \sum_{i=1}^l l_i + 1$, $\tilde{P} =\tilde{P}_k\prod_{i=1}^j P_{(j_i,k_i,l_i)}$, in which case $\deg \tilde{P} \le \deg\tilde{P}_k + \sum_{i=1}^j\deg P_{(j_i,k_i,l_i)} \le k/2 + \sum_{i=1}^j j_i/2 = j^*/2$. Once again, $(j^*,k^*,l^*)\in\A$ by construction.

Putting these terms together, we see that $\frac{\Delta E_{n;\h}}{\h}$ can be written as a sum of the terms in the form \eqref{eq:deltae-inductive}, but with $N$ replaced with $N+1$, up to $O(\h^{N+1}/\E_{n;\h}^{N+1})$, as desired.
\end{proof}
The first few coefficients which appear in the asymptotic expansion can be computed using the method above, and this is done in Appendix \ref{sec:coeff}.
\begin{remark}
Consider the case $\ell=0$. Then $S(E) = (\tau_++\tau_-)E$, $T(E) =\tau_++\tau_-$, and
\[\E_{n;\h} = S^{-1}(2\pi \h(n+1/2)) = \frac{2\pi \h(n+1/2)}{\tau_++\tau_-} = \overline\omega \h(n+1/2)\]
where $\overline\omega = 2/(\omega_+^{-1}+\omega_-^{-1})$ is the harmonic mean of the frequencies. Note then that
\[(\tau_++\tau_-)\E_{n;\h}/\h = 2\pi n + \pi\]
and hence
\[\cos(k\tau_-\E_{n;\h}/\h) = (-1)^k\cos(k\tau_+\E_{n;\h}/\h),\quad \sin(k\tau_-\E_{n;\h}/\h) = (-1)^{k+1}\sin(k\tau_+\E_{n;\h}/\h).\]
It follows that we have the simplifications
\begin{align*}
P_{(2,2,1)} &= \frac{\omega_+^2-\omega_-^2}{16}\cos(\tau_+\E_{n;\h}/\h)\\
P_{(4,4,1)} &= (\omega_+^2-\omega_-^2)\left(-\frac{5}{128}\cos(\tau_+\E_{n;\h}/\h) + \frac{1}{1024}\sin(2\tau_+\E_{n;\h}/\h)\right)\\
P_{(4,4,2)} &= \frac{\pi(\omega_++\omega_-)(\omega_+^2-\omega_-^2)}{512}\sin(2\tau_+\E_{n;\h}/\h)\\
P_{(5,5,2)} &= \frac{(\omega_+^2-\omega_-^2)^2}{512}(1+\cos(\tau_+\E_{n;\h}/\h))\\
P_{(5,11/2,3)} &= 0.
\end{align*}
Note that for $P_{(4,4,2)}$ and $P_{(5,5,2)}$ we used the double angle identities $\sin(x)\cos(x) = \frac{\sin(2x)}{2}$ and $\cos^2(x) = \frac{1+\cos(2x)}{2}$.
\end{remark}

\section{Application to the Gutzwiller Trace Formula}
\label{sec:gutz}

We now apply the asymptotics to the Gutzwiller Trace Formula. This trace formula ostensibly concerns the eigenvalue counting function
\[\sum_n \chi(E_n)\rho\left(\frac{E-E_n}{\h}\right),\]
where $\chi, \hat\rho\in C_c^\infty(\mathbb{R})$, and the sum is taken over eigenvalues of the semiclassical Schr\"odinger operator
\[P = -\h^2\Delta_g + V(x),\]
defined on some manifold $M$. Here, for $p(x,\xi) = |\xi|^2_g+V(x)$, we assume that $p^{-1}(\text{supp }\chi)$ is compact, in which case it turns out\cite{m-92} that the portion of the spectrum of $P$ in $\text{supp }\chi$ is discrete if $\h$ is small enough, allowing us to make sense of the sum as a finite sum.

This eigenvalue-counting function can be phrased in terms of the trace of a regularized Schr\"odinger propagator: indeed, by the Fourier inversion formula, we have
\begin{equation}
\label{eq:gutz-fourier}
\begin{aligned}
\sum_{n} \chi(E_n)\rho\left(\frac{E-E_n}{\h}\right) &= \sum_n \frac{\chi(E_n)}{2\pi}\int_{\mathbb{R}} e^{iEt/\h}e^{-iE_nt/\h}\hat\rho(t)\,dt \\
&=\frac{1}{2\pi}\int_{\mathbb{R}} e^{iEt/\h}\hat\rho(t)\,\text{tr}(\chi(P_\h)e^{-itP_\h/\h})\,dt.
\end{aligned}
\end{equation}

The eigenvalue-counting function turns out to relate to the classical dynamics of the classical symbol $p(x,\xi)$. Indeed, when $V(x)$ is a $C^\infty$ potential, then, under nondegenerate assumptions on the variation of periodic orbits with respect to energy, the Gutzwiller Trace Formula (see \cite{m-92}, \cite{crr-99}, etc.) reads
\begin{align*}
\sum_{n} \chi(E_n)\rho\left(\frac{E-E_n}{\h}\right) = \frac{\chi(E)}{2\pi}&\left[\h^{1-d}\left(\text{Vol}(\Sigma_E)\hat\rho(0)+O(\h)\right)\right. \\
&\left.+\sum_{\gamma}{\frac{e^{i(S_\gamma/\h+\sigma_{\gamma}\pi/2)}T_{\gamma}^\sharp}{|\det(I-\mathcal{P}_\gamma)|^{1/2}}\hat\rho(T_{\gamma})} + O(\h)\right],
\end{align*}
where the sum is taken over all closed periodic trajectories $\gamma$ of the Hamilton flow of $p(x,\xi)$
on the surface $\Sigma_E\{|\xi|^2_g+V = E\}$, $d$ the dimension of $M$, $T_{\gamma}$ is the length of $\gamma$, $T_{\gamma}^\sharp$ the primitive length of $\gamma$, $S_\gamma$ the classical action on $\gamma$, $\sigma_{\gamma}$ the Maslov index of $\gamma$, and $\mathcal{P}_\gamma$ the Poincar\'e map
restricted to the orthocomplement of the eigenspace of eigenvalue $1$.
The $O(\h)$ terms are in fact an asymptotic expansion, with coefficients depending on dynamical data, as well as derivatives of $\chi$ at the energy $E$ and $\hat\rho$ at the times $0$ or $T_\gamma$. In the case $d=1$, both leading order terms are of order $\h^0=1$, and if the energy surface $\Sigma_E$ consists of a single closed trajectory, then $\text{Vol}(\Sigma_E)$ equals the period of that trajectory.

There is ongoing work studying the behavior of the trace formula in the case when the potential is not $C^\infty$, but instead has a derivative discontinuity; in particular relating the trace to the dynamics of periodic trajectories which are allowed to \emph{reflect} at the site of singularity. The goal of this section is to give an example of this behavior, as the potentials of interest in this paper are $C^1$ but have a second derivative discontinuity.

\textbf{Notation}: Given $\psi,\chi\in C_c^\infty(\mathbb{R})$, we write
\begin{equation}
\label{eq:trcal}
\Trcal_{\psi,\chi;\h}(t) := \psi(t)\sum_{n=0}^\infty \chi(\E_{n;\h})e^{-it\E_{n;\h}/\h}.
\end{equation}

\subsection{Hamiltonian dynamics}
\label{subsec:ham}
We first study the Hamiltonian dynamics of
\[p(x,\xi) = \frac{\xi^2}{2m}+V(x).\]
That is, we look for integral curves $(x(t),\xi(t))$ satisfying
\[\dot x(t) = \frac{\xi(t)}{m},\quad \dot\xi(t) = -V'(x(t)) = \begin{cases} -m\omega_-^2x(t) & x(t)<0 \\ 0 & 0<x(t)<\ell \\ -m\omega_+^2(x(t)-\ell) & x(t)>\ell\end{cases}.\]
Since $V'(x)$ is continuous, given an initial $(x(0),\xi(0))$, there exists a unique $C^1$ Hamiltonian trajectory. This will be referred to as a \emph{transmitted} trajectory. Note that such a trajectory will in fact be $C^\infty$ away from where $x(t) = 0$ or $x(t) = \ell$; near those points, $x(t)$ is $C^{2,1}$ and $\xi(t)$ is $C^{1,1}$.

We will also consider trajectories which have \emph{reflection}: in that case, we require $x(t)$ to be continuous, $(x(t),\xi(t))$ to be $C^\infty$ away from where $x(t) = 0$ or $\ell$, and if $x(t_0) = 0$ or $\ell$, then $\xi(t)$ is piecewise $C^\infty$ to the left and right of $t_0$, with $\xi(t_0-) = -\xi(t_0+)$, where $\xi(t_0\pm)$ are the limits of $\xi$ as $t\to t_0^{\pm}$.

We now consider such trajectories. In the region where $x(t)<0$, the solution must be given by
\[x(t) = -\frac{\sqrt{2E/m}}{\omega_-}\sin\left(\omega_-(t-t_0)\right),\quad \xi(t) = -\sqrt{2Em}\cos\left(\omega_-(t-t_0)\right)\]
for some $t_0$, where $E$ is the (constant) value of $p(x,\xi) = \frac{\xi^2}{2m}+V(x)$ along the trajectory. Notice that such a trajectory stays in $\{x<0\}$ for $t\in(t_0,t_0 + \tau_-)$ where $\tau_- = \pi/\omega_-$, and $x(t)\to 0$ as $t$ approaches either endpoint. In particular, a trajectory stays for a time of $\tau_-$ in this region before exiting or reflecting (note that this time is independent of $E$). In addition, over this trajectory we have that the classical action is given by
\[
S = \int \xi\,dx = \int_0^{\tau_-} \left(-\sqrt{2Em}\cos(\omega_-t)\right)\left(-\sqrt{2E/m}\cos(\omega_-t)\right)\,dt = \tau_-E.
\]
Similarly, in the region $0<x(t)<\ell$, the solution must be given by
\[x(t) = \pm\sqrt{\frac{2E}{m}}(t-t_0),\quad \xi(t) = \pm\sqrt{2Em}\]
for some $t_0$ and some choice of sign $\pm$. In that case, the trajectory stays in $\{0<x<\ell\}$ for $t\in \left(t_0,t_0+\frac{\sqrt{m}\ell}{\sqrt{2E}}\right)$ (if $+$) or $t\in\left(t_0-\frac{\sqrt{m}\ell}{\sqrt{2E}},t_0\right)$ (if $-$), and $x(t)$ reaches the endpoint $0$ or $\ell$ as $t$ approaches either of its endpoints. In particular, a trajectory stays for a time of $\frac{\sqrt{m}\ell}{\sqrt{2E}}$ in this region before exiting or reflecting. The action is
\[S = \pm\sqrt{2Em}\cdot\pm\ell = \sqrt{2Em}\ell.\]
Finally, in the region $x(t)>\ell$, the solution is given by
\[x(t) = \ell + \frac{\sqrt{2E/m}}{\omega_+}\sin\left(\omega_-(t-t_0)\right),\quad \xi(t) = \sqrt{2Em}\cos\left(\omega(t-t_0)\right)\]
for some $t_0$. Similar to the first case, the trajectory stays for a time of $\tau_+ = \pi/\omega_+$ before exiting or reflecting, and the action is $\tau_+E$.

Note that a transmitted trajectory making one full revolution is given by one trajectory in the region $x<0$, one trajectory traveling forward in $0<x<\ell$, one trajectory in $x>\ell$, and one trajectory traveling backwards in $0<x<\ell$. Thus the period for a transmitted trajectory with one revolution is given by
\[T = \tau_- + \tau_+ + 2\frac{\sqrt{m}\ell}{\sqrt{2E}} = \frac{\sqrt{2m}\ell}{\sqrt{E}} + \tau_- + \tau_+;\]
note this is exactly the formula in \eqref{eq:t}. Similarly, the total action is
\[S = \tau_+E + \tau_-E + 2\sqrt{2Em}\ell,\]
which is the formula in \eqref{eq:s}.

\begin{figure}[h]
\centering
\begin{tikzpicture}[scale=0.7]
\begin{axis}[xmin=-3.5,xmax=6.5, ymin=-2.5, ymax = 2.5, axis lines = center, xlabel = {$x$}, ylabel = {$\xi$}, xtick = {-2,3,4}, xticklabels = { , , }, 
xtick style = {very thick}, extra x ticks = {0}, extra x tick labels = { }, ytick = {-2,2}, ytick style = {very thick}, yticklabels = { , }
]
\addplot[domain=-90:90]({-2*cos(x)},{2*sin(x)});
\addplot[domain=0:3]({x},{2});
\addplot[domain=-90:90]({3+cos(x)},{2*sin(x)});
\addplot[domain=0:3]({x},{-2});
\addplot[domain=-2.5:2.5,thick,dashed]({0},{x});
\addplot[domain=-2.5:2.5,thick,dashed]({3},{x});
\node at (axis cs:-2.8,-0.38) {$-\frac{\sqrt{2E}}{\omega_-}$};
\node at (axis cs:0.3,-0.4) {$0$};
\node at (axis cs:3.35,-0.38) {$\ell$};
\node at (axis cs:5,-0.4) {$\ell+\frac{\sqrt{2E}}{\omega_+}$};
\node at (axis cs:-1,-2.2) {$-\sqrt{2E}$};
\node at (axis cs:-0.8,2.2) {$\sqrt{2E}$};
\end{axis}
\end{tikzpicture}
\caption{The energy surface $\frac{\xi^2}{2}+V(x) = E$.}
\end{figure}
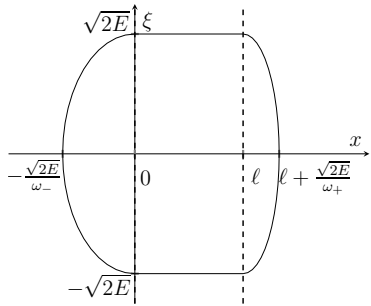
We now describe the periodic trajectories with exactly one reflection. Note that such trajectories must reflect at either $x=0$ or $x=\ell$; moreover, at the point of reflection, the trajectory must stay on one side of that point. We will thus classify these trajectories by the point and direction of the reflection, denoted $0^{\pm}$ and $\ell^{\pm}$ (with $+$ if the trajectory stays on the right and $-$ if it stays on the left).

For trajectories with reflection point and direction $0^-$, we note that such trajectories must complete a half-period in the quadratic end $x<0$, reflect, complete some number (possibly zero) of full revolutions, before starting on the half-period again. This means that the period of this periodic trajectory is $\tau_-$ (for the half-period on the left end), plus some multiple of full periods $T(E)$, i.e.\ the period is $\tau_-+kT(E)$ for some $k\in\mathbb{Z}$. Moreover, if we define the action of the trajectory as the sum of the actions on the piecewise smooth pieces, then by following the above descriptions we find that the action is given by $kS(E) + \tau_-E$.

For trajectories with reflection point and direction $0^+$, we note that such trajectories must traverse the flat portion from $x=\ell$ to $x=0$, reflect, complete some number (possibly zero) of full revolutions, traverse the flat portion from $x=0$ to $x=\ell$, and complete the half-period in the quadratic end $x>\ell$, before starting on the flat portion from $x=\ell$ to $x=0$ again. This means that the total period is
\[\frac{\sqrt{m}\ell}{\sqrt{2E}} + k^*T(E) + \frac{\sqrt{2m}\ell}{\sqrt{2E}} + \tau_+ =  kT(E)-\tau_-\]
for some integers $k^*$, $k$, where $k = k^*-1$. The total action is
\[\sqrt{2mE}\ell + k^*S(E) + \sqrt{2mE}\ell + \tau_+E = kS(E)-\tau_-E.\]
The cases of reflecting at $\ell^{\pm}$ can be analyzed similarly. The results are summarized in the table in Figure \ref{fig:reflect}. We note that in each case the period is of the form $T_{k,\alpha,\beta}(E)$ for some $\alpha,\beta\in\mathbb{Z}$ with $|\alpha|+|\beta|=1$, while the corresponding action is $S_{k,\alpha,\beta}(E)$.

\begin{figure}[h]
\label{fig:reflect}
\begin{tabular}{c|c|c}
Reflection point/direction & Period & Action \\\hline
$0^-$ & $kT(E)+\tau_-$ & $kS(E)+\tau_-E$ \\\hline
$0^+$ & $kT(E)-\tau_+$ & $kS(E)-\tau_+E$ \\\hline
$\ell^-$ & $kT(E)-\tau_-$ & $kS(E)-\tau_-E$ \\\hline
$\ell^+$ & $kT(E)+\tau_+$ & $kS(E)+\tau_+E$
\end{tabular}
\caption{The four types of periodic trajectories with one reflection. Note that if $\ell=0$, then the middle two options do not apply.}
\end{figure}

\subsection{Analysis of a related sum via Poisson Summation}

To study the regularized trace $\sum_n\chi(E_{n;\h})e^{-itE_{n;\h}/\h}$, we first study the corresponding sum replacing $E_{n;\h}$ by $\E_{n;\h}$, i.e.\ sums of the form
\[\sum{\chi(\E_{n;\h})e^{-it\E_{n;\h}/\h}},\]
with $\chi\in C_c^\infty(\mathbb{R}_+)$ and $\hat\rho\in C_c^\infty(\mathbb{R})$. We recall Poisson Summation Formula. If we set the Fourier transform conventions
\[\mathcal{F}f(\xi) = \int_{\mathbb{R}}{e^{-i\xi x}f(x)\,dx},
\quad\mathcal{F}^{-1}\varphi(x) = \frac{1}{2\pi}\int_{\mathbb{R}}{e^{ix\xi}\varphi(\xi)\,d\xi},\]
then the Poisson Summation Formula gives
\[\sum_{n\in\mathbb{Z}}\varphi(n) = 2\pi\sum_{k\in\mathbb{Z}}\mathcal{F}^{-1}\varphi(2\pi k)\]
(cf.\ \cite[Equation (7.2.1)']{h1} for the exact formula given these conventions). Applying this to $\varphi(\xi) = \chi(S^{-1}(2\pi \h(\xi+1/2)))e^{-itS^{-1}(2\pi \h(\xi+1/2))/\h}$, we have
\begin{align*}
2\pi\mathcal{F}^{-1}\varphi(2\pi k) &= \int_{\mathbb{R}}{e^{i(2\pi k)\xi}\chi(S^{-1}(2\pi \h(\xi+1/2)))e^{-itS^{-1}(2\pi \h(\xi+1/2))/\h}\,d\xi} \\
&=\int_{\mathbb{R}} e^{ikS(\eta)/\h - ik\pi}\chi(\eta)e^{-it\eta/\h}\,\frac{S'(\eta)}{2\pi\h}\,d\eta \\
&=\frac{(-1)^k}{2\pi \h}\int_{\mathbb{R}} e^{i(kS(\eta)-t\eta)/\h}\chi(\eta)S'(\eta)\,d\eta,
\end{align*}
where $\eta = S^{-1}(2\pi\h(\xi+1/2))$, i.e.\ $\xi = \frac{S(\eta)}{2\pi\h}-\frac{1}{2}$. Thus, combining the two equations, and recalling $\E_{n;\h} = S^{-1}(2\pi\h(n+1/2))$, we obtain
\begin{equation}
\label{eq:poisson-trace}
\sum_{n=0}^\infty \chi(\E_{n;\h})e^{-it\E_{n;\h}/\h} = \sum_{n\in\mathbb{Z}} \chi(\E_{n;\h})e^{-it\E_{n;\h}/\h} = \sum_{k\in\mathbb{Z}} \frac{(-1)^k}{2\pi h}\int_{\mathbb{R}} e^{i(kS(\eta)-t\eta)/\h}\chi(\eta)S'(\eta)\,d\eta.
\end{equation}
We note that the first equality follows from the support assumptions on $\chi$. Below, we will use the method of stationary phase (\cite[Theorem 7.7.5]{h1}), which gives:
\[\int u(x)e^{if(x)/\h}\,dx \sim \frac{(2\pi \h)^{n/2}e^{if(x_0)/\h}e^{i\frac{\pi}{4}\text{sgn }D^2f(x_0)}}{|\det D^2f(x_0)|^{1/2}}\sum_{j\ge 0}\h^jL_ju(x_0),\]
where $L_j$ is a differential operator of order $2j$, $L_0 = 1$.

We can now state our results. Recalling the classical action
\[S(E) = 2\sqrt{2m}\ell\sqrt{E} + (\tau_++\tau_-)E\]
and the period of one orbit,
\[T(E) = S'(E) = \frac{\sqrt{2m}\ell}{\sqrt{E}} + \tau_++\tau_-.\]
\begin{lemma}
\label{lem:poisson-mathcale}
Suppose $\chi,\hat\rho\in C_c^\infty(\mathbb{R})$. Then
\[\mathcal{F}^{-1}(\Trcal_{\hat\rho,\chi;\h})(E/\h) = \frac{1}{2\pi}\int_{\mathbb{R}} e^{iEt/\h}\hat\rho(t)\left(\sum \chi(\E_{n;\h})e^{-it\E_{n;\h}/\h} \right)\,dt \sim \sum_{j\ge 0}a_j\h^j,\]
where
\[a_0 = \sum_{k\in\mathbb{Z}}\frac{(-1)^ke^{ikS(E)/\h}}{2\pi}\chi(E)T(E)\hat\rho(kT(E)),\]
and $a_j$ is a sum of terms consisting of the product of derivatives of $\chi$ evaluated at $E$ and derivatives of $\hat\rho$ evaluated at $kT(E)$, with $k\in\mathbb{Z}$. In particular, if $\mathcal{T}_0(E)\cap\text{supp }\hat\rho = \emptyset$, then the integral is $O(\h^\infty)$.
\end{lemma}
\begin{proof}
By \eqref{eq:poisson-trace},
\begin{equation}
\label{eq:motley1}
\frac{1}{2\pi}\int_{\mathbb{R}} e^{iEt/\h}\hat\rho(t)\left(\sum_{n\in\mathbb{Z}} \chi(\E_{n;\h})e^{-it\E_{n;\h}/\h} \right)\,dt
=\sum_{k\in\mathbb{Z}}\frac{(-1)^k}{4\pi^2 \h} \int_{\mathbb{R}^2}e^{i(kS(\eta) + Et- t\eta)/\h}\chi(\eta)S'(\eta)\hat\rho(t)\,dt\,d\eta.
\end{equation}
For the phase function $\Phi(t,\eta) = kS(\eta)+Et-t\eta$, we have
\[\partial_t\Phi = E-\eta,\quad \partial_\eta\Phi = kS'(\eta)-t,\]
so the stationary point is at $\eta= E$ and $t = kS'(E) = kT(E)$. Furthermore,
\[D^2\Phi = \begin{pmatrix} 0 & -1 \\ -1 & kS''(\eta)\end{pmatrix},\]
so $\det D^2\Phi = 1$ and $\text{sgn }D^2\Phi = 0$.
Hence, for each integral in the sum, we have
\[\int_{\mathbb{R}^2}e^{i(kS(\eta) + Et- t\eta)/\h}\chi(\eta)S'(\eta)\hat\rho(t)\,dt\,d\eta \sim 2\pi \h e^{ikS(E)/\h}\sum_{j=0}^\infty \h^jL_j|_{(t,\eta) = (kT(E),E)}[\chi(\eta)S'(\eta)\hat\rho(t)].\]
Note that the first term in the asymptotic sum is
\[\chi(E)S'(E)\hat\rho(kT(E)) = \chi(E)T(E)\hat\rho(kT(E)),\]
while the remaining terms consist of product of derivatives of $\chi$ evaluated at $E$ and derivatives of $\hat\rho$ evaluated at $kT(E)$. Plugging this into each term in the sum in \eqref{eq:motley1} gives the desired result.
\end{proof}

\subsection{Application to the regularized trace}
We now Taylor expand. Noting that
\[\chi(E_{n;\h})=\sum_{j=0}^N\frac{\chi^{(j)}(\E_{n;\h})}{j!}(\Delta E_{n;\h})^j + O(|\Delta E_{n;\h}|^{N+1}),\]
with the last term $O(\h^{3(N+1)}\E_{n;\h}^{-2(N+1)})$, and
\begin{align*}
e^{-itE_{n;\h}/\h} &= \sum_{k=0}^N\frac{(-it)^k}{k!}\left(\frac{\Delta E_{n;\h}}{\h}\right)^k e^{-it\E_{n;\h}/\h} + O\left(\left|\frac{\Delta E_{n;\h}}{\h}\right|^{N+1}\right), 
\end{align*}
with the last term $O(\h^{2(N+1)}\E_{n;\h}^{-2(N+1)})$, it follows by multiplying the above equations that
\[\chi(E_{n;\h})e^{-it\E_{n;\h}/\h} = \sum_{j,k=0}^N  \frac{(-it)^k}{j!k!}\h^j\chi^{(j)}(\E_{n;\h})\left(\frac{\Delta E_{n;\h}}{\h}\right)^{j+k}e^{-it\E_{n;\h}/\h} + O(\h^{2(N+1)}\E_{n;\h}^{-2(N+1)}).\]
Summing the above equation, we note that the sum of the error term (over $n$ with $E_{n;\h}\in\text{supp }\chi$) can be crudely bounded by $O(\h^{2N+1})$ since there are $O(\h^{-1})$ terms. Hence,
\begin{equation}
\label{eq:trace-expand}
\begin{aligned}
\sum_n \chi(E_{n;\h}) e^{-it\E_{n;\h}/\h} 
&= \sum_{j,k=0}^N \frac{(-it)^k}{j!k!}\h^j\sum_{n=0}^\infty \chi^{(j)}(\E_{n;\h})\left(\frac{\Delta E_{n;\h}}{\h}\right)^{j+k}e^{-it\E_{n;\h}/\h} + O(\h^{2N+1}).
\end{aligned}
\end{equation}
Note that each power of $\frac{\Delta E_{n;\h}}{\h}$ can be written as a sum of trigonometric functions in $\E_{n;\h}/\h$ multiplied by smooth functions in $\E_{n;\h}$. We study the behavior of the sum of such terms below.

\begin{lemma}
\label{lem:tr-deltaem}
Let $m\in\mathbb{N}$ and $\chi\in C_c^\infty(\mathbb{R}_+)$. Then, there exists a family of $C_c^\infty(\mathbb{R}_+)$ functions $\{\chi_{j,\alpha,\beta}\}$, indexed over $j,\alpha,\beta\in\mathbb{Z}$, $j\ge 2m$, $|\alpha|+|\beta|\le j/2$, such that for any $N\in\mathbb{N}$, we have the asymptotic expansion
\[\sum_{n=0}^\infty \chi(\E_{n;\h})\left(\frac{\Delta E_{n;\h}}{\h}\right)^me^{-it\E_{n;\h}/\h} = \sum_{j=2m}^{N-1} \h^j\sum_{\substack{\alpha,\beta\in\mathbb{Z} \\ |\alpha|+|\beta|\le j/2}} \Trcal_{1,\chi_{j,\alpha,\beta};\h}(t-\alpha\tau_--\beta\tau_+) + O(\h^N)\]
(where $\Trcal$ is defined in \eqref{eq:trcal}). The $O(\h^N)$ error term is locally uniform in $t$. Moreover, the functions $\chi_{2m,\alpha,\beta}$ satisfy $\chi_{2m,\alpha,\beta}(E) = \chi(E)c_{2m,\alpha,\beta}(E)$, where $c_{2m,\alpha,\beta}$ satisfy satisfy
\[\sum_{|\alpha|+|\beta|\le m}c_{2m,\alpha,\beta}(E)e^{i\alpha x}e^{i\beta y} = \left(\frac{1}{E^2T(E)}P_{2,2,1}(x,y)\right)^m.\]
\end{lemma}
\begin{proof}
By Theorem \ref{thm:e-asymp}, we can write
\[\frac{\Delta E_{n;\h}}{\h} = \sum_{j=2}^{N-1} \h^j P_j\left(\frac{\tau_-\E_{n;\h}}{\h},\frac{\tau_+\E_{n;\h}}{\h};\E_{n;\h}\right) + O(\h^N)\]
for any $N$, where
\[P_j(x,y;E) = \sum_{(k,l)\:\,(j,k,l)\in\mathcal{A}} \frac{1}{E^kT(E)^l}P_{(j,k,l)}(x,y)\]
is a two-variable trigonometric polynomial in $(x,y)$, whose coefficients are smooth functions of $E$ (in the region $E>\epsilon$), satisfying $\deg P_j\le j/2$. It follows that we can write
\[\left(\frac{\Delta E_{n;\h}}{\h}\right)^m = \sum_{j=2m}^{N-1} \h^j \tilde{P}_j\left(\frac{\tau_-\E_{n;\h}}{\h},\frac{\tau_+\E_{n;\h}}{\h};\E_{n;\h}\right) + O(\h^N)\]
for some $\tilde{P}_j(x,y;E)$ satisfying the same properties; note in particular that $\tilde{P}_{2m}(x,y;E) = \left(\frac{1}{E^2T(E)}P_{(2,2,1)}(x,y)\right)^m$. If we write
\[\tilde{P}_j(x,y;E) = \sum_{|\alpha|+|\beta|\le j} c_{j,\alpha,\beta}(E)e^{i\alpha x}e^{i\beta y},\]
then we have
\begin{equation}
\label{eq:tr-deltae-exp}
\begin{aligned}
&\sum_{n=0}^\infty \chi(\E_{n;\h})\left(\frac{\Delta E_{n;\h}}{\h}\right)^me^{-it\E_{n;\h}/\h} \\
&= \sum_{j=2m}^N\h^j\left(\sum_{n=0}^\infty \chi(\E_{n;\h})\tilde{P}_j\left(\frac{\tau_-\E_{n;\h}}{\h},\frac{\tau_+\E_{n;\h}}{\h};\E_{n;\h}\right)e^{-it\E_{n;\h}/\h}\right) \\
&+ \sum_{n=0}^\infty \chi(\E_{n;\h})e^{-it\E_{n;\h}}O(\h^{N+1});
\end{aligned}
\end{equation}
note that the last sum is $O(\h^N)$ by a crude counting argument as there are $O(\h^{-1})$ terms with $\E_{n;\h}\in\text{supp }\chi$. We note that
\begin{align*}
&\sum_{n=0}^\infty\chi(\E_{n;\h})\tilde{P}_j\left(\frac{\tau_-\E_{n;\h}}{\h},\frac{\tau_+\E_{n;\h}}{\h};\E_{n;\h}\right)e^{-it\E_{n;\h}/\h} \\
&= \sum_{|\alpha|+|\beta|\le j/2}\sum_{n=0}^\infty \chi(\E_{n;\h})c_{j,\alpha,\beta}(\E_{n;\h})e^{i\alpha\tau_-\E_{n;\h}/\h}e^{i\beta\tau_+\E_{n;\h}/\h}e^{-t\E_{n;\h}/\h} \\
&= \sum_{|\alpha|+|\beta|\le j/2}\sum_{n=0}^\infty \chi_{j,\alpha,\beta}(\E_{n;\h})e^{-i(t-\alpha\tau_--\beta\tau_+)/\h} = \sum_{|\alpha|+|\beta|\le j}\Trcal_{1,\chi_{j,\alpha,\beta};\h}(t-\alpha\tau_--\beta\tau_+),
\end{align*}
where $\chi_{j,\alpha,\beta} = \chi(E)c_{j,\alpha,\beta}(E)$. Plugging this into \eqref{eq:tr-deltae-exp} yields the desired result.
\end{proof}

\begin{lemma}
\label{lem:gutz-deltae}
For any $l\in\mathbb{N}$, we have
\[\frac{1}{2\pi}\int_{\mathbb{R}} e^{iEt/\h}(-it)^l\hat\rho(t)\left(\sum \chi(\E_{n;\h})\left(\frac{\Delta E_{n;\h}}{\h}\right)^me^{-it\E_{n;\h}/\h} \right)\,dt \sim \sum_{j\ge 2m}b_j\h^j,\]
where
\[b_{2m} =\frac{\chi(E)T(E)}{2\pi}\sum_{k\in\mathbb{Z}}i^{-(2k+l)}e^{iS_{k,\alpha,\beta}(E)/\h}c_{2m,\alpha,\beta}(E)T_{k,\alpha,\beta}(E)^l\hat\rho(T_{k,\alpha,\beta}(E)) 
,\]
with $c_{2m,\alpha,\beta}$ the functions in Lemma \ref{lem:tr-deltaem}, and $b_j$ consists of products of derivatives of $\chi$ evaluated at $E$ and derivatives of $\hat\rho$ evaluated at $T_{k,\alpha,\beta}(E)$, where $k,\alpha,\beta\in\mathbb{Z}$ with $|\alpha|+|\beta|\le j/2$.
\end{lemma}
\begin{proof}
By Lemma \ref{lem:tr-deltaem}, we have
\begin{align*}
&(-it)^l\hat\rho(t)\sum \chi(\E_{n;\h})\left(\frac{\Delta E_{n;\h}}{\h}\right)^me^{-it\E_{n;\h}/\h} \\
&= \sum_{j=2m}^{N-1} \h^j(-it)^l\hat\rho(t)\sum_{\substack{\alpha,\beta\in\mathbb{Z} \\ |\alpha|+|\beta|\le j/2}} \Trcal_{1,\chi_{j,\alpha,\beta};\h}(t-\alpha\tau_--\beta\tau_+) + O(\h^N) \\
&= \sum_{j=2m}^{N-1} \h^j\sum_{\substack{\alpha,\beta\in\mathbb{Z} \\ |\alpha|+|\beta|\le j/2}} \Trcal_{\psi_{\alpha,\beta},\chi_{j,\alpha,\beta};\h}(t-\alpha\tau_--\beta\tau_+) + O(\h^N)
\end{align*}
where $\psi_{\alpha,\beta}(t) = (-i(t+\alpha\tau_-+\beta\tau_+))^l\hat\rho(t+\alpha\tau_-+\beta\tau_+)$. Thus
\begin{align*}
&\frac{1}{2\pi}\int_{\mathbb{R}} e^{iEt/\h}(-it)^l\hat\rho(t)\left(\sum \chi(\E_{n;\h})\left(\frac{\Delta E_{n;\h}}{\h}\right)^me^{-it\E_{n;\h}/\h} \right)\,dt \\
&=\frac{1}{2\pi}\int_{\mathbb{R}} e^{iEt/\h}\sum_{j=2m}^{N-1} \h^j\sum_{\substack{\alpha,\beta\in\mathbb{Z} \\ |\alpha|+|\beta|\le j/2}} \Trcal_{\psi_{\alpha,\beta},\chi_{j,\alpha,\beta};\h}(t-\alpha\tau_--\beta\tau_+) + O(\h^N)\,dt \\
&=\sum_{j=2m}^{N-1} \h^j\sum_{\substack{\alpha,\beta\in\mathbb{Z} \\ |\alpha|+|\beta|\le j/2}}\frac{1}{2\pi}\int_{\mathbb{R}} e^{iE(t+\alpha\tau_-+\beta\tau_+)/\h} \Trcal_{\psi_{\alpha,\beta},\chi_{j,\alpha,\beta};\h}(t)\,dt + O(\h^N) \\
&=\sum_{j=2m}^{N-1} \h^j\sum_{\substack{\alpha,\beta\in\mathbb{Z} \\ |\alpha|+|\beta|\le j/2}}e^{i(\alpha\tau_-+\beta\tau_+)E/\h}\mathcal{F}^{-1}(\Trcal_{\psi_{\alpha,\beta},\chi_{j,\alpha,\beta};\h})(E/\h) + O(\h^N).
\end{align*}
By Lemma \ref{lem:poisson-mathcale}, each term $\mathcal{F}^{-1}(\Trcal_{\psi_{\alpha,\beta},\chi_{j,\alpha,\beta};\h})(E/\h)$ can be written as an asymptotic series $\sum_{j\ge 0}a_j\h^j$, where the coefficients involve derivatives of $\chi_{j,\alpha,\beta}$ evaluated at $E$ times derivatives of $\psi_{\alpha,\beta}$ evaluated at $kT(E)$, $k\in\mathbb{Z}$. Since $\chi_{j,\alpha,\beta}(E) = \chi(E)c_{j,\alpha,\beta}(E)$, it follows that derivatives of $\chi_{j,\alpha,\beta}$ at $E$ can be written in terms of the derivatives of $\chi$ at $E$, while derivatives of $\psi_{\alpha,\beta}(t) = (-i(t+\alpha\tau_-+\beta\tau_+))^l\hat\rho(t+\alpha\tau_-+\beta\tau_+)$, at $t=kT(E)$, can be written in terms of derivatives of $\hat\psi$, but at the time $t = kT(E)+\alpha\tau_-+\beta\tau_+ = T_{k,\alpha,\beta}(E)$. This shows the claimed asymptotic expansion.

Moreover, the leading coefficient is
\begin{align*}
&\sum_{k\in\mathbb{Z}} \frac{(-1)^ke^{ikS(E)/\h}}{2\pi}\chi_{j,\alpha,\beta}(E)T(E)\psi_{\alpha,\beta}(kT(E)) \\
&= \frac{\chi(E)T(E)}{2\pi}\sum_{k\in\mathbb{Z}}i^{-(2k+l)}e^{ikS(E)/\h}c_{j,\alpha,\beta}(E)T_{k,\alpha,\beta}(E)^l\hat\rho(T_{k,\alpha,\beta}(E)).
\end{align*}
Hence, the leading order term (of order $\h^{2m}$) has the coefficient
\begin{align*}
&\sum_{|\alpha|+|\beta|\le m} e^{i(\alpha\tau_-+\beta\tau_+)E/\h}\frac{\chi(E)T(E)}{2\pi}\sum_{k\in\mathbb{Z}}i^{-(2k+l)}e^{ikS(E)/\h}c_{2m,\alpha,\beta}(E)T_{k,\alpha,\beta}(E)^l\hat\rho(T_{k,\alpha,\beta}(E))\\
&=\frac{\chi(E)T(E)}{2\pi}\sum_{k\in\mathbb{Z}}i^{-(2k+l)}e^{iS_{k,\alpha,\beta}(E)/\h}c_{2m,\alpha,\beta}(E)T_{k,\alpha,\beta}(E)^l\hat\rho(T_{k,\alpha,\beta}(E)),
\end{align*}
as claimed.
\end{proof}

\begin{proof}[Proof of Theorem \ref{thm:gutz-qual}]
By Equation \eqref{eq:trace-expand}, we have
\[\sum_{n=0}^\infty \chi(E_{n;\h})e^{-itE_{n;\h}/\h} = \sum_{m=0}^N\sum_{l=0}^m \h^{m-l}(-it)^l\tilde\chi_{m,l}(\E_{n;\h})\left(\frac{\Delta E_{n;\h}}{\h}\right)^me^{-it\E_{n;\h}/\h} + O(\h^{2N+1})\]
where $\tilde\chi_{m,l}(E) = \frac{\chi^{(m-l)}(E)}{l!(m-l)!}$. Plugging this expansion into \eqref{eq:gutz-fourier}, we see by Lemma \ref{lem:gutz-deltae} that the $(m,l)$ term contributes an asymptotic expansion of the form $\h^{m-l}\sum_{j\ge 2m}b_j\h^j$, where the $b_j$ coefficient is a sum of products which include $\hat\rho$ evaluated at $T_{k,\alpha,\beta}$, $k\in\mathbb{Z}$ and $|\alpha|+|\beta|\le j/2$, i.e.\ at times in $\mathcal{T}_{\lfloor j/2\rfloor}(E)$. In particular, the terms for which $j\ge 2(n+1)$ contribute $O(\h^{2(N+1)})$ terms, while the remaining terms involve $\hat\rho$ evaluated at times in $\mathcal{T}_N$. Hence, if $\text{supp }\hat\rho\cap\mathcal{T}_N(E)=\emptyset$, then the coefficients of all powers of $h$ lower than $\h^{2(N+1)}$ vanish, i.e.\ the overall sum is $O(\h^{2(N+1)})$.
\end{proof}
\begin{proof}[Proof of Theorem \ref{thm:gutz-quant}]
If $T^*\in\mathcal{T}_0(E)$, i.e.\ $T^* = kT(E)$ for $k\in\mathbb{Z}$, then the result follows from Lemma \ref{lem:poisson-mathcale}. For $T^*\in\mathcal{T}_1(E)\backslash\mathcal{T}_0(E)$, we note that we need the contribution of the term $\sum_{n=0}^\infty(-it)\chi(\E_{n;\h})\frac{\Delta E_{n;\h}}{\h}e^{-it\E_{n;\h}/\h}$. We use the expansion
\[\frac{\Delta E_{n;\h}}{\h} = \frac{\h^2}{\E_{n;\h}^2T(\E_{n;\h})}P_{(2,2,1)}\left(\frac{\tau_-\E_{n;\h}}{\h},\frac{\tau_+\E_{n;\h}}{\h}\right) + O(\h^4/\E_{n;\h}^4).\]
Noting that
\[P_{(2,2,1)}(x,y) = \frac{1}{32}\left(\omega_-^2(e^{ix}+e^{-ix})+\omega_+^2(e^{iy}+e^{-iy})\right),\]
we have that
\[\frac{1}{E^2T(E)}P_{(2,2,1)}(x,y) = \sum_{|\alpha|+|\beta|=1}c_{2,\alpha,\beta}(E)e^{i\alpha x}e^{i\beta y}\]
for
\[c_{2,\pm 1,0}(E) = \frac{\omega_-^2}{32E^2T(E)},\quad c_{2,0,\pm 1}(E) = \frac{\omega_+^2}{32E^2T(E)}.\]
By Lemma \ref{lem:gutz-deltae}, the sum is $b_2\h^2+O(\h^3)$, with
\begin{align*}
b_2 = \frac{\chi(E)T(E)}{2\pi}\sum_{k\in\mathbb{Z}}i^{-(2k+1)}e^{iS_{k,\alpha,\beta}(E)/\h}c_{2,\alpha,\beta}(E)T_{k,\alpha,\beta}(E)\hat\rho(T_{k,\alpha,\beta}(E)).
\end{align*}
Writing $T^* = kT(E) + \alpha\tau_- + \beta\tau_+$, we see that the corresponding coefficient $c_{2,\alpha,\beta}(E)$ is $\frac{\omega_{\pm}^2}{32E^2T(E)}$, with the $\pm$ being $-$ if $|\alpha|=1$ and $+$ if $|\beta|=1$. It follows that
\[b_2 = \frac{\chi(E)T(E)}{2\pi}i^{-(2k+1)}e^{iS_{k,\alpha,\beta}(E)/\h}\frac{\omega_{\pm}^2}{32E^2T(E)}T^*\hat\rho(T^*).\]
Thus the sum overall is
\[\h^2\frac{\chi(E)e^{iS_{k,\alpha,\beta}(E)/\h}\omega_{\pm}^2T^*}{64\pi E^2}\hat\rho(T^*) + O(\h^3),\]
as claimed.
\end{proof}
\begin{remark}
By the table in Figure \ref{fig:reflect}, we see that a reflected trajectory with period $T^* = T_{k,\alpha,\beta}(E)$ also has an action given by $S_{k,\alpha,\beta}(E)$, which appears in the coefficient above.
\end{remark}

\section{Heat Trace}
\label{sec:heat}
We now study the heat trace. For simplicity we will consider the case $\ell=0$, so the potential is an asymmetric harmonic oscillator with no flat region. In that case, the leading order asymptotic $\E_{n;\h}$ is affine in $n$:
\[\E_{n;\h} = \frac{2\pi \h(n+1/2)}{\tau_++\tau_-} = \overline\omega \h(n+1/2),\]
where $\overline\omega = \frac{2}{\omega_-^{-1}+\omega_+^{-1}} = \frac{2\pi}{\tau_++\tau_-}$ is the harmonic mean of $\omega_+$ and $\omega_-$. Moreover, $\E_{n;h}$ is linear in $\h$ as well, and in fact $E_{n;\h}$ is linear in $\h$ as well: indeed, the equation for the eigenvalue reduces to
\[\tilde\theta\left(\frac{z}{2\omega_-}\right) + \tilde\theta\left(\frac{z}{2\omega_+}\right) = \pi n,\quad z = E/\h.\]
Hence, we will take $\h=1$, write $E_n = E_{n;1}$ and $\E_n = \E_{n;1}$, and consider the asymptotics as $t\to 0$. To get $\h$-dependent asymptotics, one can replace $t$ by $\h t$ in all results below. 

We first note the computation
\begin{equation}
\label{eq:heat-mathcale}
\sum_{n=0}^\infty e^{-t\E_n} = \sum_{n=0}^\infty e^{-t\overline\omega(n+1/2)} = \frac{e^{-t\overline\omega/2}}{1-e^{-t\overline\omega}} = \frac{1}{2\sinh(\overline{\omega} t/2)},
\end{equation}
which in fact holds for any $t\in\mathbb{C}$ with $\text{Re}(t)>0$. Moreover, for $t>0$, we have the asymptotic expansion
\begin{equation}
\label{eq:sinhexp}
\frac{1}{2\sinh(\overline{\omega}t/2)} = 
\frac{1}{\overline\omega t} - \frac{1}{24}\overline\omega t + \frac{7}{5760}(\overline\omega t)^3 + \dots.
\end{equation}
We first observe the following asymptotics which follow from \eqref{eq:heat-mathcale}.
\begin{lemma}
\label{lem:heat-em}
Let $m\ge 1$ be an integer.
\begin{enumerate}
\item If $m\ge 2$ and $a_n$ is a sequence with $|a_n|\le C\E_{n}^{-m}$, then
\[\sum_{n=0}^\infty e^{-t\E_n}a_n\]
is $C^{m-2}$ in $t$ as $t\to 0^+$; in particular it admits an asymptotic expansion $c_0+c_1t+\dots+c_{m-2}t^{m-2} + o(t^{m-2})$.
\item If $a_n = \E_{n}^{-m}$, then we have
\[\sum_{n=0}^\infty e^{-t\E_n}\E_n^{-m} = \frac{(-1)^m}{(m-1)!\overline\omega }t^{m-1}\log(t) + \varphi_m(t)\]
for some smooth function $\varphi_m$.
\end{enumerate}
\end{lemma}
\begin{proof}
For (i), we observe that differentiating the sum $m-2$ times yields the function
\[\sum_{n=0}^\infty e^{-t\E_n} a_n(-\E_n)^{m-2},\]
which is continuous as $t\to 0^+$ since $\sum_{n=0}^\infty|a_n(-\E_n)^{m-2}| \le C\sum_{n=0}^\infty \E_n^{-2}<\infty$ as $\E_n$ is affine in $n$. For (ii), we differentiate the sum $m$ times to yield
\[\sum_{n=0}^\infty (-1)^me^{-t\E_n} = \frac{(-1)^m}{2\sinh(\overline\omega t/2)};\]
the latter equals $\frac{(-1)^m}{\overline\omega t}$ up to a function which is smooth as $t\to 0^+$. The result then follows from integrating $m$ times, noting the integration formula
\[\int t^k\log(t)\,dt = \frac{t^{k+1}\log(t)}{k+1} - \frac{t^{k+1}}{(k+1)^2} + C.\]
\end{proof}
To analyze the heat trace, particularly terms appearing from the error $\Delta E_n$ from the main term, we first note the following observation:
\begin{lemma}
\label{lem:heat-trig}
Suppose $\tau\in\mathbb{R}$ does not belong to $\frac{2\pi}{\overline\omega}\mathbb{Z} = (\tau_++\tau_-)\mathbb{Z}$. Then for any non-negative integer $m$ we have that the functions
\[\sum_{n=0}^\infty e^{-t\E_n}\cos\left(\tau\E_n\right)\E_n^{-m},\quad \sum_{n=0}^\infty e^{-t\E_n}\sin\left(\tau\E_n\right)\E_n^{-m}\]
are smooth as $t\to 0^+$.
\end{lemma}
\begin{proof}
For $m=0$, the sum can be written as
\[\sum_{\pm}c_{\pm}\sum_{n=0}^\infty e^{-t\E_n}e^{\pm i\tau\E_n}\]
where $(c_+,c_-) = (1/2,1/2)$ or $(1/(2i),-1/(2i))$ for the cosine and sine cases, respectively. Using \eqref{eq:heat-mathcale}, we have
\[\sum_{n=0}^\infty e^{-t\E_n}e^{\pm i\tau\E_n} = \frac{1}{2\sinh\left(\frac{\overline\omega}{2}\left(t \mp i\tau\right)\right)} = \frac{1}{2\sinh\left(\frac{\overline\omega t}{2}\mp i\frac{\overline\omega \tau}{2}\right)}\text{ for }t>0.\]
The zeros of $\sinh$ are at $i\pi\mathbb{Z}$, and by assumption, we have $\frac{\overline\omega\tau}{2}\not\in \pi\mathbb{Z}$. Hence the above function is smooth as $t\to 0^+$.

For $m\ne 0$, the result follows by differentiating the sum $m$ times.
\end{proof}
We are now in a position to analyze the heat trace
$\sum_{n=0}^\infty e^{-tE_n}$.
We write $e^{-tE_n} = e^{-t\E_n}e^{-t\Delta E_n}$, and we Taylor expand\footnote{We note that $t\Delta E_n$ is uniformly bounded since we are interested in the regime $t\to 0^+$, and that $\Delta E_n\to 0$ as $n\to\infty$; hence the constant in the $O(t^3(\Delta E_n)^3)$ can be made uniform.} the last term to obtain
\[e^{-tE_n} = e^{-t\E_n}\left(1-t\Delta E_n + \frac{t^2}{2}(\Delta E_n)^2 + O\left(t^3(\Delta E_n)^3\right)\right).\]
Thus
\[\sum_{n=0}^\infty e^{-tE_n} = \sum_{j=0}^2\sum_{n=0}^\infty e^{-t\E_n}\frac{(-t\Delta E_n)^j}{j!} + \sum_{n=0}^\infty e^{-t\E_n}O(t^3(\Delta E_n)^3).\]
The last term is $t^3$ times $\sum_{n=0}^\infty e^{-t\E_n}O(\E_n^{-6})$; the latter sum is $C^4$ by Lemma \ref{lem:heat-em}. Thus overall the term admits an asymptotic expansion
\[\tilde{c}_3t^3 + \tilde{c}_4t^4 + \dots + \tilde{c}_7t^7 + o(t^7)\]
for some constants $\tilde{c}_3,\dots,\tilde{c}_7$.

For the main terms, the $j=0$ term is explicitly $\frac{1}{2\sinh(\overline\omega t/2)}$ and admits an expansion \eqref{eq:sinhexp}.
For the other terms, we use
\begin{align*}
\Delta E_n &= \frac{\omega_+^2-\omega_-^2}{16\E_n^2T(\E_n)}\cos(\tau_+\E_n) +\frac{\omega_+^2-\omega_-^2}{\E_n^4T(\E_n)}\left(-\frac{5}{128}\cos(\tau_+\E_{n}) + \frac{1}{1024}\sin(2\tau_+\E_{n})\right)\\
&+\frac{\pi(\omega_++\omega_-)(\omega_+^2-\omega_-^2)}{512\E_n^4T(\E_n)^2}\sin(2\tau_+\E_n) +\frac{(\omega_+^2-\omega_-^2)^2}{512\E_n^5T(\E_n)^2}(1+\cos(2\tau_+\E_n)) + O(\E_n^{-6}).
\end{align*}
and hence
\[(\Delta E_n)^2 = \left(\frac{\omega_+^2-\omega_-^2}{16\E_n^2T(\E_n)}\cos(\tau_+\E_n)\right)^2 + O(\E_n^{-6}) = \frac{(\omega_+^2-\omega_-^2)^2}{512\E_n^4T(\E_n)^2}(1+\cos(2\tau_+\E_n)) + O(\E_n^{-6}).\]
In analyzing the $j=1$ term
\[-t\sum_{n=0}^\infty e^{-t\E_n}\Delta E_n\]
we note that the $O(\E_n^{-6})$ term in $\Delta E_n$ gives a sum which is $C^4$, so admits an asymptotic expansion up to $t^4$; multiplying by $t$ gives an asymptotic expansion up to $t^5$. We also note that $\tau_+\not\in(\tau_++\tau_-)\mathbb{Z}$, and that if $2\tau_+\in(\tau_++\tau_-)\mathbb{Z}$, then we'd necessarily have $2\tau_+ = \tau_++\tau_-$ since $0<2\tau_+<2(\tau_++\tau_-)$; however by assumption we have $\omega_+\ne\omega_-$, so we have $2\tau_+\not\in(\tau_++\tau_-)\mathbb{Z}$ as well. Hence, by Lemma \ref{lem:heat-trig}, the other terms in $\Delta E_n$ contribute smooth expansions as $t\to 0$, except for the term
\[\frac{(\omega_+^2-\omega_-^2)^2}{512\E_n^5T(\E_n)^2}.\]
By Lemma \ref{lem:heat-em} applied to $m=5$,
\[-t\sum_{n=0}^\infty e^{-t\E_n}\frac{(\omega_+^2-\omega_-^2)^2}{512\E_n^5T(\E_n)^2} = t\varphi(t) + \frac{1}{24\overline\omega}\frac{(\omega_+^2-\omega_-^2)^2}{512T^2}t\log(t)\]
where $\varphi$ is smooth. It follows that
\[-t\sum_{n=0}^\infty e^{-t\E_n}\Delta E_n = \sum_{j=1}^5 a_jt^j + \frac{1}{24\overline\omega}\frac{(\omega_+^2-\omega_-^2)^2}{512T^2}t\log(t) + o(t^5)\]
for some $a_1,\dots,a_5$. Similarly, in analyzing the $j=2$ term
\[\frac{t^2}{2}\sum_{n=0}^\infty e^{-t\E_n}(\Delta E_n)^2,\]
we note that the $O(\E_n^{-6})$ term contributes a sum with asymptotic expansion up to $t^4$ (hence up to $t^6$ after multiplying by $t^2$), while the oscillatory term also contributes a smooth expansion. By Lemma \ref{lem:heat-em} with $m=4$, the nonoscillatory term contributes
\[\frac{t^2}{2}\sum_{n=0}^\infty e^{-t\E_n}\frac{(\omega_+^2-\omega_-^2)^2}{512\E_n^4T(\E_n)^2} = t^2\varphi(t) + \frac{(\omega_+^2-\omega_-^2)^2}{2\cdot 512 T^2}\frac{1}{6\overline\omega}t^5\log(t).\]
Thus,
\[\frac{t^2}{2}\sum_{n=0}^\infty e^{-t\E_n}(\Delta E_n)^2 = \sum_{j=2}^6 b_jt^j + \frac{(\omega_+^2-\omega_-^2)^2}{2\cdot 512 T^2}\frac{1}{6\overline\omega}t^5\log(t) + o(t^6)\]
for some $b_2,\dots,b_6$. Putting all of this together, we have that
\begin{align*}
\sum_{n=0}^\infty e^{-tE_n} &= \frac{1}{\overline\omega t} + \sum_{j=1}^4 c_jt^j + \left(\frac{1}{24\overline\omega}\frac{(\omega_+^2-\omega_-^2)^2}{512T^2} + \frac{(\omega_+^2-\omega_-^2)^2}{2\cdot 512 T^2}\frac{1}{6\overline\omega}\right)t^5\log(t) + O(t^5) \\
&=\frac{1}{\overline\omega t} +\sum_{j=1}^4 c_jt^j + \frac{(\omega_+^2-\omega_-^2)^2}{4096\overline\omega T^2}t^5\log(t)+ O(t^5)
\end{align*}
for some $c_1,\dots,c_4$, thus proving Theorem \ref{thm:heat}.

\appendix

\section{Proofs of Lemmas}
\label{sec:lemma-proofs}

\begin{proof}[Proof of Lemma \ref{lem:tildethetaasymp}]

We recall from \eqref{eq:tildetheta} that $\begin{pmatrix} \cos(\tilde\theta(z)) \\ \sin(\tilde\theta(z))\end{pmatrix}$ is proportional to the vector
\[\begin{pmatrix} \frac{z^{1/2}}{\Gamma(3/4-z)} \\ -\frac{1}{\Gamma(1/4-z)}\end{pmatrix}.\]
From the functional equation
\[\Gamma(z)\Gamma(1-z) = \frac{\pi}{\sin(\pi z)},\]
we have $\frac{1}{\Gamma(1-z)} = \pi^{-1}\Gamma(z)\sin(\pi z)$. Hence
\begin{align*}
\begin{pmatrix} \frac{z^{1/2}}{\Gamma\left(\frac{3}{4}-z\right)} \\ -\frac{1}{\Gamma\left(\frac{1}{4}-z\right)} \end{pmatrix} &= \frac{1}{\pi}\begin{pmatrix} z^{1/2}\Gamma\left(\frac{1}{4}+z\right)\sin\left(\pi\left(\frac{1}{4}+z\right)\right) \\ -\Gamma\left(\frac{3}{4}+z\right)\sin\left(\pi\left(\frac{3}{4}+z\right)\right)\end{pmatrix} \\
&= \frac{z^{1/2}\Gamma\left(\frac{1}{4}+z\right)}{\pi}\begin{pmatrix}\sin\left(\pi\left(\frac{1}{4}+z\right)\right) \\ -F(z)\sin\left(\pi\left(\frac{3}{4}+z\right)\right)\end{pmatrix},
\end{align*}
where
\begin{equation}
\label{eq:F}
F(z) = \frac{\Gamma(z+3/4)}{\Gamma(z+1/4)z^{1/2}}.
\end{equation}
Noting that we can write $\sin\left(\pi\left(\frac{1}{4}+z\right)\right) = \cos\left(\pi\left(z-\frac{1}{4}\right)\right)$ and $-\sin\left(\pi\left(\frac{3}{4}+z\right)\right) = {\sin\left(\pi\left(z-\frac{1}{4}\right)\right)}$, it follows that
\[\begin{pmatrix} \cos(\tilde\theta(z)) \\ \sin(\tilde\theta(z)) \end{pmatrix} = \tilde{r}(z)\begin{pmatrix} \cos\left(\pi\left(z-\frac{1}{4}\right)\right) \\ F(z)\sin\left(\pi\left(z-\frac{1}{4}\right)\right)\end{pmatrix}\]
for some $\tilde{r}(z)>0$. Thus, whenever $\cos(\pi(z-1/4))\ne 0$, we can write
\begin{equation}
\label{eq:tans}
\tan(\tilde\theta(z)) = F(z)\tan(\pi(z-1/4)).
\end{equation}
We now note that, for any $a,b\in\mathbb{C}$, we have an asymptotic expansion
\[\frac{\Gamma(z+a)}{\Gamma(z+b)}\sim z^{a-b}\left(1 + \frac{C_1}{z} + \frac{C_2}{z^2} + \dots\right)\quad\text{ as }z\to\infty, \quad\text{Re}(z)>0;\]
see \cite[Equation 6.1.47]{as} for the first two coefficients, as well as \cite{et-51}  or \cite{ln-12} for the proof of a full asymptotic expansion.
Moreover, the coefficients $C_n$ are given by
\[C_n = \binom{a-b}{n}B_n^{(a-b+1)}(a),\]
where $B_n^\alpha(x)$ is the so-called ``generalized Bernoulli polynomial'' or ``N{\o}rlund polynomial'' (see \cite{n-24}) satisfying
\[\left(\frac{t}{e^t-1}\right)^{\alpha}e^{xt} = \sum_{n=0}^\infty B_n^\alpha(x)\frac{t^n}{n!}.\]
For $a=3/4$ and $b=1/4$, one can check that $C_n=0$ for all odd $n$, with $C_2 = 1/64$ and $C_4 = -19/8192$. It follows that
\begin{equation}
\label{eq:F-asymp}
F(z) = 1 + \frac{1}{64z^2} - \frac{19}{8192z^4} + \dots.
\end{equation}
To prove monotonicity of $\tilde\theta$, we use the following property of $F$:
\begin{lemma}
For $z>0$, the function $\frac{F'(z)}{F(z)}$ is negative and increasing.
\end{lemma}
\begin{proof}
We recall the digamma function $\psi(z) = \frac{\Gamma'(z)}{\Gamma(z)}$. It can be written\cite[Equation 6.3.16]{as} as the infinite series
\[\psi(1+z) = -\gamma + \sum_{n=1}^\infty\frac{z}{n(n+z)}\implies \psi(z) = -\gamma + \sum_{n=0}^\infty \left[\frac{1}{n+1}-\frac{1}{n+z}\right],\]
where $\gamma$ is the Euler-Mascheroni constant. Noting by logarithmic differentiation that
\[\frac{F'(z)}{F(z)}   = \psi(z+3/4)-\psi(z+1/4) - \frac{1}{2z}\]
and using the infinite series above, as well as noting the telescopic series
\[\frac{1}{z} = \sum_{n=0}^\infty \left[\frac{1}{n+z} - \frac{1}{n+z+1}\right],\]
we see that
\[\frac{F'(z)}{F(z)} = \sum_{n=0}^\infty\left[\frac{1}{n+z+1/4}-\frac{1}{n+z+3/4} - \frac{1/2}{n+z} + \frac{1/2}{n+z+1}\right].\]
One can check that the terms inside can be algebraically rewritten to yield
\[\frac{F'(z)}{F(z)} = \sum_{n=0}^\infty -\frac{3}{32}\left(\frac{1}{(n+z)(n+z+1/4)(n+z+3/4)(n+z+1)}\right),\]
from which it is clear that $\frac{F'(z)}{F(z)}$ is negative and increasing.
\end{proof}
We now show that $\tilde\theta$ is strictly increasing. From \eqref{eq:tans}, it suffices to show that $F(z)\tan\left(\pi\left(z-\frac{1}{4}\right)\right)$ is increasing in intervals where it is well-defined. Noting that
\begin{align*}
&\frac{d}{dz}\left(F(z)\tan\left(\pi\left(z-\frac{1}{4}\right)\right)\right) \\
&= F'(z)\tan\left(\pi\left(z-\frac{1}{4}\right)\right) + \pi F(z)\sec^2\left(\pi\left(z-\frac{1}{4}\right)\right) \\
&= F(z)\sec^2\left(\pi\left(z-\frac{1}{4}\right)\right)\left(\frac{F'(z)}{F(z)}\sin\left(\pi\left(z-\frac{1}{4}\right)\right)\cos\left(\pi\left(z-\frac{1}{4}\right)\right) + \pi\right),
\end{align*}
and noting that $F(z)$ and $\sec^2(z)$ are always positive, it now suffices to show 
\[\frac{F'(z)}{F(z)}\sin(\pi(z-1/4))\cos(\pi(z-1/4)) + \pi>0\]
for all $z$. For $0< z < 1/4$, this follows since both $\frac{F'(z)}{F(z)}$ and ${\sin\left(\pi\left(z-\frac{1}{4}\right)\right)\cos\left(\pi\left(z-\frac{1}{4}\right)\right)}$ are negative, and hence the overall sum is positive. For $z\ge 1/4$, we note that
\[\frac{F'(1/4)}{F(1/4)} = \psi(1) - \psi(1/2) - \frac{1}{2(1/4)} = \log(4) - 2\approx -0.614.\]
Since $\frac{F'(z)}{F(z)}$ is negative and increasing, it follows that $\left|\frac{F'(z)}{F(z)}\right|\le\left|\frac{F'(1/4)}{F(1/4)}\right|<1$ for $z\ge 1/4$. Combined with the bound $|\sin(\theta)\cos(\theta)| = \frac{1}{2}|\sin(2\theta)|\le \frac{1}{2}$ for any $\theta$, it follows that for $z\ge 1/4$ we have
\[\frac{F'(z)}{F(z)}\sin(\pi(z-1/4))\cos(\pi(z-1/4)) + \pi > \pi -  1\cdot\frac{1}{2} > 0,\]
as desired. This shows that $\tilde\theta$ is strictly increasing.

Moreover, since $-\frac{z^{1/2}}{\Gamma(3/4-z)}$ and $\frac{1}{\Gamma(1/4-z)}$ have zeros at $\frac{3}{4}+\mathbb{N}$ and $\frac{1}{4}+\mathbb{N}$, respectively, it follows that $\tilde\theta(z)$ must equal $0$, respectively $\pi/2$, modulo $\pi$ precisely when the fractional part of $z$ equals $1/4$, respectively $3/4$. Combined with the convention $\tilde\theta(0) = -\pi/2$ and the monotonicity of $\tilde\theta$, it follows that we must have
\[\tilde\theta(z) = \pi(z-1/4)\text{ when }z\in\left(\frac{1}{4}+\mathbb{N}\right)\cup\left(\frac{3}{4}+\mathbb{N}\right) = \frac{1}{4}+\frac{1}{2}\mathbb{N}.\]
Since $\tilde\theta(z)$ and $\pi(z-1/4)$ are both strictly increasing, and they agree at $\frac{1}{4}+\frac{1}{2}\mathbb{N}$, with both functions increasing by $\pi/2$ between each such consecutive pair of points, it follows that their difference must always be at most $\pi/2$. Thus, writing $\tilde\theta(z) = \pi(z-1/4) + r(z)$, we must have $|r(z)|<\pi/2$ for all $z>0$.

To get an asymptotic series for $r(z)$, we use \eqref{eq:tans} and the tangent subtraction formula $\tan(a-b) = \frac{\tan(a)-\tan(b)}{1+\tan(a)\tan(b)}$. We then have
\begin{align*}
\tan(\tilde\theta(z)-\pi(z-1/4)) &= \frac{(F(z)-1)\tan(\pi(z-1/4))}{1+F(z)\tan^2(\pi(z-1/4))} \\
&= \frac{(F(z)-1)\sin(\pi(z-1/4))\cos(\pi(z-1/4))}{\cos^2(\pi(z-1/4))+F(z)\sin^2(\pi(z-1/4))}.
\end{align*}
Writing
\[\sin(\pi(z-1/4))\cos(\pi(z-1/4)) = \frac{1}{2}\sin(2\pi z-\pi/2) = -\frac{1}{2}\cos(2\pi z)\]
and
\begin{align*}
\cos^2(\pi(z-1/4))+F(z)\sin^2(\pi(z-1/4)) &= 1+(F(z)-1)\sin^2(\pi(z-1/4)) \\
&= 1+(F(z)-1)\frac{1-\cos(2\pi z-\pi/2)}{2} \\
&= 1+\frac{F(z)-1}{2}(1-\sin(2\pi z)),
\end{align*}
it follows that
\begin{align*}
\tan\left(\tilde\theta(z)-\pi\left(z-\frac{1}{4}\right)\right) &= -\frac{F(z)-1}{2}\frac{\cos(2\pi z)}{1+\frac{F(z)-1}{2}(1-\sin(2\pi z))} \\
&= -\sum_{j=0}^\infty\left(\frac{F(z)-1}{2}\right)^{j+1}\cos(2\pi z)(\sin(2\pi z)-1)^j
\end{align*}
(note that this sum is absolutely convergent for $z$ sufficiently large). Using \eqref{eq:F-asymp}, we see that each term $\left(\frac{F(z)-1}{2}\right)^{j+1}$ is an asymptotic series in $1/z$ starting at $2(j+1)$, and it is multiplied by $\cos(2\pi z)(\sin(2\pi z)-1)^j$, which is a trigonometric polynomial of degree $2j+1$ in $2\pi z$. Hence, $\tan\left(\tilde\theta(z)-\pi\left(z-\frac{1}{4}\right)\right)$ is of the desired form claimed in Lemma \ref{lem:tildethetaasymp}.
The desired result follows by applying $\arctan$ to both sides, noting that $\arctan$ is analytic near $z=0$.
\end{proof}
We remark that the computation of the asymptotic series for $r(z)$ up to $N=5$ follows by taking the $j=0$ and $j=1$ terms above, using from \eqref{eq:F-asymp} that
\[\frac{F(z)-1}{2} = \frac{1}{128z^2} - \frac{19}{16384z^4} + O(z^{-6}).\]

\section{Computation of the asymptotic expansion coefficients}
\label{sec:coeff}

We compute the first few coefficients appearing in the asymptotic expansion, using the proof of Theorem \ref{thm:e-asymp} in Section \ref{sec:e-asymp}. From the proof, we already have
\[\frac{\Delta E_{n;\h}}{\h} = \frac{\h^2}{16\E_{n;\h}^2T(\E_{n;\h})}\left(\omega_-^2\cos\left(\frac{\tau_-\E_{n;\h}}{\h}\right) + \omega_+^2\cos\left(\frac{\tau_+\E_{n;\h}}{\h}\right)\right) + O(\h^4/\E_{n;\h}^4).\]
Using \eqref{eq:delta-e-recursive}, we compute some of the remaining terms. The term
\[-\frac{2}{T(\E_{n;\h})}\frac{\h^{j-1}S^{(j)}(\E_{n;\h})}{2j!}\left(\frac{\Delta E_{n;\h}}{\h}\right)^j,\quad j=2\]
contributes to the $(5,11/2,3)$ term in the expansion, up to $O(\h^7\E_{n;\h}^{-15/2})$. Using that $S''(E) = -\frac{\sqrt{m/2}\ell}{E^{3/2}}$, this contribution is
\begin{align*}
&-\frac{2}{T(\E_{n;\h})}\frac{\h S''(\E_{n;\h})}{4}\cdot \left(\frac{\h^2}{16\E_{n;\h}^2T(\E_{n;\h})}\left(\omega_-^2\cos\left(\frac{\tau_-\E_{n;\h}}{\h}\right) + \omega_+^2\cos\left(\frac{\tau_+\E_{n;\h}}{\h}\right)\right)\right)^2 \\
&=\frac{\h^5}{\E_{n;\h}^{11/2}T(\E_{n;\h})^3}\frac{\sqrt{m}\ell}{512\sqrt{2}}\left(\omega_-^2\cos\left(\frac{\tau_-\E_{n;\h}}{\h}\right) + \omega_+^2\cos\left(\frac{\tau_+\E_{n;\h}}{\h}\right)\right)^2 \\
&=\frac{\h^5}{\E_{n;\h}^{11/2}T(\E_{n;\h})^3}P_{(5,11/2,3)}\left(\frac{\tau_-\E_{n;\h}}{\h},\frac{\tau_+\E_{n;\h}}{\h}\right),
\end{align*}
where
\[P_{(5,11/2,3)}(x,y) = \frac{\sqrt{m}\ell}{1024\sqrt{2}}\left(\omega_-^2\cos(x)+\omega_+^2\cos(y)\right)^2.\]
We now consider terms of the form
\[-\frac{2}{T(\E_{n;\h})}\frac{r^{(j)}\left(\frac{\E_{n;\h}}{2\h\omega_\pm}\right)}{j!(2\omega_\pm)^j}\left(\frac{\Delta E_{n;\h}}{\h}\right)^j.\]
For $j\ge 2$, this is $O(\h^6/\E_{n;\h}^6)$. The $j=0$ term contributes to the $(2,2,1)$ and $(4,4,1)$ terms in the expansion, up to $O(\h^6/\E_{n;\h}^6)$. Using that $r(z) = -\frac{\cos(2\pi z)}{128z^2} + \frac{5\cos(2\pi z)}{4096z^4} - \frac{\sin(4\pi z)}{32768z^4}+ O(z^{-6})$, the $(2,2,1)$ contribution is
\[\frac{\h^2\omega_{\pm}^2\cos(\tau_\pm\E_{n;\h}/\h)}{16\E_{n;\h}^2T(\E_{n;\h})};\]
adding these over $\pm$ gives the leading order term. The $(4,4,1)$ contribution is
\[\frac{\h^4\omega_{\pm}^4}{\E_{n;\h}^4T(\E_{n;\h})}\left(-\frac{5\cos(\tau_\pm\E_{n;\h}/\h)}{128} + \frac{\sin(2\tau_\pm\E_{n;\h}/\h)}{1024}\right).\]
It follows that
\[P_{(4,4,1)}(x,y) = -\frac{5\omega_-^4\cos(x)+5\omega_+^4\cos(y)}{128} + \frac{\omega_-^4\sin(2x)+\omega_+^4\sin(2y)}{1024}.\]
For $j=1$, using that $r'(z) = \frac{-\pi\sin(2\pi z)}{64z^2} - \frac{\cos(2\pi z)}{64z^3} + O(z^{-4})$, we have that the $j=1$ term contributes to the $(4,4,2)$ and $(5,5,2)$ terms in the expansion. The $(4,4,2)$ contribution is
\begin{align*}
&-\frac{2}{T(\E_{n;\h})}\frac{1}{2\omega_{\pm}}\frac{-\pi \h^2\omega_{\pm}^2\sin(\tau_\pm\E_{n;\h}/\h)}{16\E_{n;\h}^2}\left(\frac{\h^2}{\E_{n;\h}^2T(\E_{n;\h})}P_{(2,2,1)}\left(\frac{\tau_-\mathcal{E}_{n;\h}}{\h},\frac{\tau_+\mathcal{E}_{n;\h}}{\h}\right)\right) \\
&=\frac{\h^4\omega_{\pm}}{\E_{n;\h}^4T(\E_{n;\h})^2}\frac{\pi\sin(\tau_\pm\E_{n;\h}/\h)}{256}\left(\omega_-^2\cos\left(\frac{\tau_-\E_{n;\h}}{\h}\right) + \omega_+^2\cos\left(\frac{\tau_+\E_{n;\h}}{\h}\right)\right)
\end{align*}
and hence
\[P_{(4,4,2)}(x,y) = \frac{\pi}{256}(\omega_-\sin(x)+\omega_+\sin(y))(\omega_-^2\cos(x)+\omega_+^2\cos(y)).\]
The $(5,5,2)$ contribution is
\begin{align*}
&-\frac{2}{T(\E_{n;\h})}\frac{1}{2\omega_{\pm}}\frac{- \h^3\omega_{\pm}^3\cos(\tau_{\pm}\E_{n;\h}/\h)}{16\E_{n;\h}^2}\left(\frac{\h^2}{\E_{n;\h}^2T(\E_{n;\h})}P_{(2,2,1)}\left(\frac{\tau_-\mathcal{E}_{n;\h}}{\h},\frac{\tau_+\mathcal{E}_{n;\h}}{\h}\right)\right) \\
&=\frac{\h^5\omega_{\pm}^2}{\E_{n;\h}^4T(\E_{n;\h})^2}\frac{\cos(\tau_\pm\E_{n;\h}/\h)}{256}\left(\omega_-^2\cos\left(\frac{\tau_-\E_{n;\h}}{\h}\right) + \omega_+^2\cos\left(\frac{\tau_+\E_{n;\h}}{\h}\right)\right)
\end{align*}
and hence
\[P_{(5,5,2)}(x,y) = P_{(2,2,1)}(x,y)^2 = \left(\frac{1}{16}(\omega_-^2\cos(x)+\omega_+^2\cos(y))\right)^2.\]

\bibliographystyle{plain}
\bibliography{bathtub}

\end{document}